\documentclass{amsart}
\usepackage[margin=1in]{geometry}
\usepackage{color,graphicx,mathrsfs}

\def\eq{\begin{equation}}
\def\endeq{\end{equation}}
\def\bbm{\begin{bmatrix}}
\def\ebm{\end{bmatrix}}
\def\bpm{\begin{pmatrix}}
\def\epm{\end{pmatrix}}
\def\bvm{\begin{vmatrix}}
\def\evm{\end{vmatrix}}
\def\d{\text{d}}

\newtheorem{rhp}{Riemann-Hilbert Problem}

\newtheorem{theorem}{Theorem}

\newtheorem{lemma}{Lemma}
\theoremstyle{definition}
\newtheorem{definition}{Definition}
\newtheorem{remark}{Remark}

\makeatletter
\@addtoreset{equation}{section}
\makeatother

\makeatletter
\renewcommand*\env@matrix[1][c]{\hskip -\arraycolsep
  \let\@ifnextchar\new@ifnextchar
  \array{*\c@MaxMatrixCols #1}}
\makeatother

\title[Large-degree asymptotics of rational Painlev\'e-IV functions]{Large-degree asymptotics of rational Painlev\'e-IV functions associated to generalized Hermite polynomials}
\author{Robert J. Buckingham}
\address[R. J. Buckingham]{Department of Mathematical Sciences\\ University of Cincinnati\\ PO Box 210025\\ Cincinnati, OH 45221.}
\email{buckinrt@uc.edu}
\urladdr{http://homepages.uc.edu/~buckinrt/}

\begin{document}
\begin{abstract}
The Painlev\'e-IV equation has three families of rational solutions generated 
by the generalized Hermite polynomials.  Each family is indexed by two positive 
integers $m$ and $n$.  These functions have applications to nonlinear wave 
equations, random matrices, fluid dynamics, and quantum mechanics.  Numerical 
studies suggest the zeros and poles form a deformed $n\times m$ rectangular 
grid.  Properly scaled, the zeros and poles appear to densely fill 
certain curvilinear rectangles as $m,n\to\infty$ with $r:=m/n$ a fixed 
positive real number.  Generalizing a method of Bertola and Bothner 
\cite{Bertola:2014} used to study rational Painlev\'e-II functions, we 
express the generalized Hermite rational Painlev\'e-IV functions in terms 
of certain orthogonal polynomials on the unit circle.  Using the Deift-Zhou 
nonlinear steepest-descent method, we asymptotically analyze the associated 
Riemann-Hilbert problem in the limit $n\to\infty$ with $m=r\cdot n$ for 
$r$ fixed.  We obtain an explicit characterization of the 
boundary curve and determine the leading-order asymptotic expansion of the 
functions in the pole-free region.
\end{abstract}
\maketitle

\section{Introduction}

Rational solutions of the Painlev\'e-IV equation
\eq
w_{yy}=\frac{(w_y)^2}{2w}+\frac{3}{2}w^3+4yw^2+2(y^2-\alpha)w+\frac{\beta}{w}, \quad w:\mathbb{C}\to\mathbb{C} \text{ with parameters } \alpha,\beta\in\mathbb{C}
\label{p4}
\endeq
arise in the study of steady-state distributions of electric charges for a 
two-dimensional Coulomb gas in a parabolic potential \cite{Marikhin:2001};  
rational solutions of the defocusing nonlinear Schr\"odinger equation 
\cite{Clarkson:2006}, the Boussinesq equation \cite{Clarkson:2008}, the 
classical Boussinesq system \cite{Clarkson:2009a}, and the point vortex 
equations with quadrupole background flow \cite{Clarkson:2009b}; 
rational-logarithmic solutions of the dispersive water wave equation and the 
modified Boussinesq equation \cite{Clarkson:2009c}; rational extensions of the 
harmonic oscillator and related exceptional orthogonal polynomials 
\cite{Marquette:2013,MarquetteQ:2016};  and the recurrence coefficients for 
polynomials 
orthogonal to the weight $e^{-x^2}|x|^n$ and Gaussian Unitary Ensemble matrices 
with repeated eigenvalues \cite{ChenF:2006}.  The fact that these functions 
have interesting mathematical properties in their own right is suggested by 
plots of the zeros and poles.  Indeed, as $\alpha$ and $\beta$ vary along 
certain sequences, the zeros and poles (when appropriately scaled) appear to 
form strikingly regular patterns in the complex plane that densely fill out 
curvilinear rectangles (for the rational functions that can be expressed in 
terms of generalized Hermite polynomials;  see Figures 
\ref{Hmn-zeros1}--\ref{Hmn-zeros2}) and curvilinear rectangles 
with equilateral curvilinear triangles attached to the edges (for the rational 
solutions expressed in terms of generalized Okamoto polynomials) 
\cite{Clarkson:2003}.  In this work we explicitly determine the boundary curves 
for the rational Painlev\'e-IV functions associated to the generalized Hermite 
polynomials, and derive the leading-order asymptotic expansions of these 
rational functions in the exterior of the zero/pole region.  

\begin{figure}[h]
\includegraphics[width=2.1in]{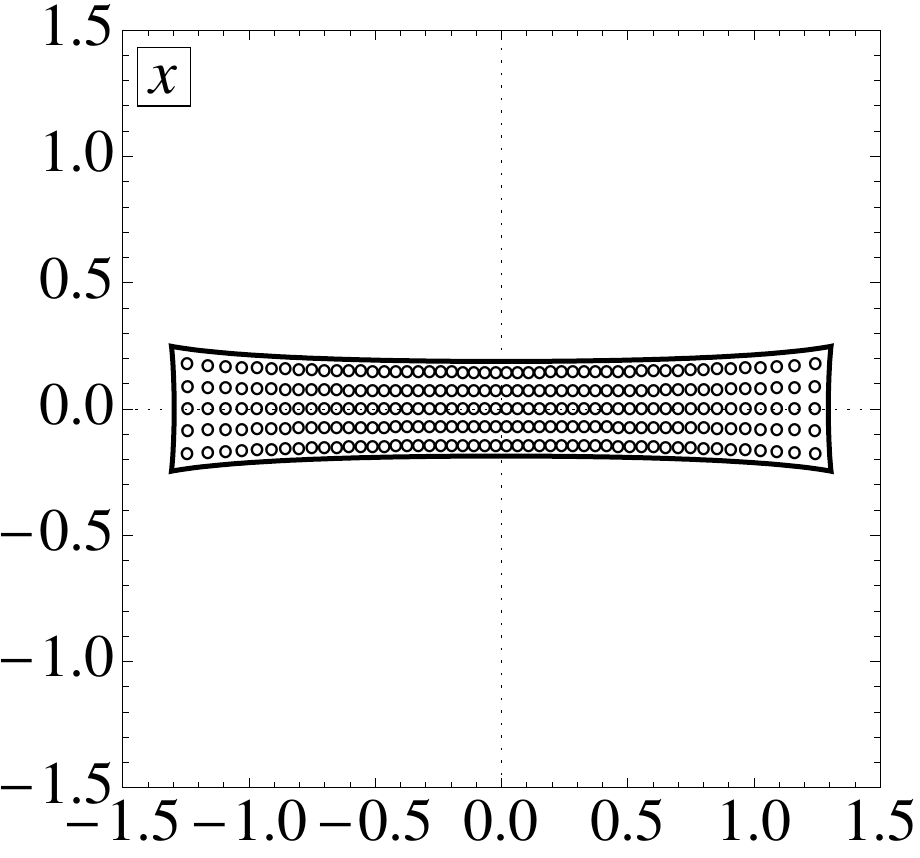}
\includegraphics[width=2.1in]{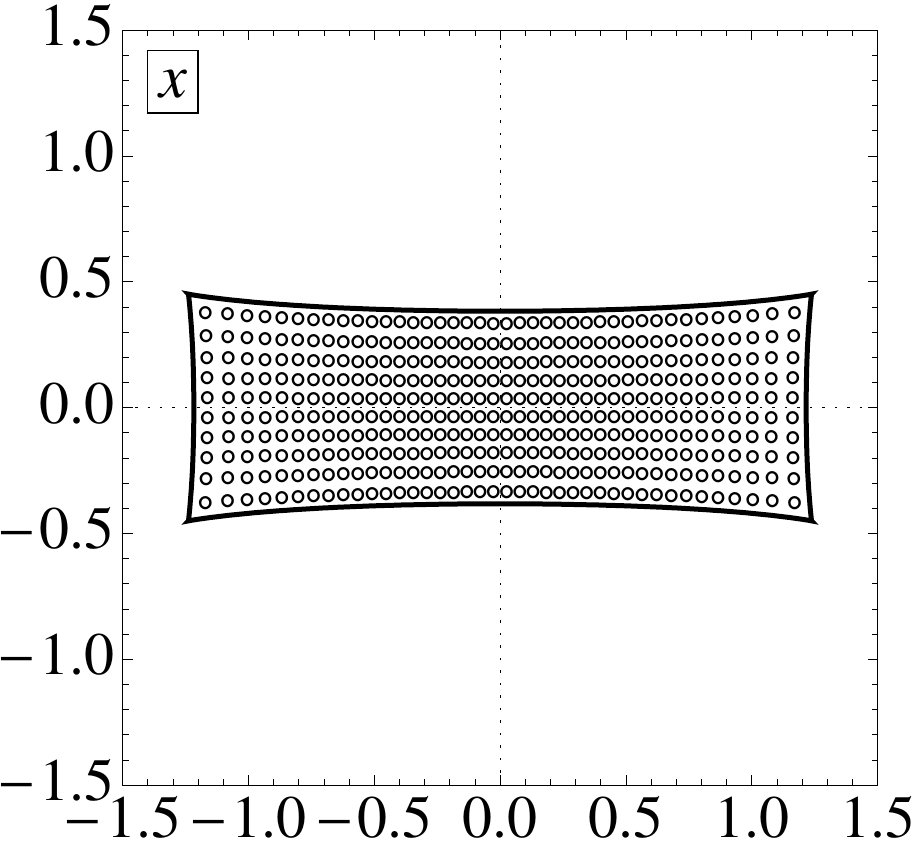}
\includegraphics[width=2.1in]{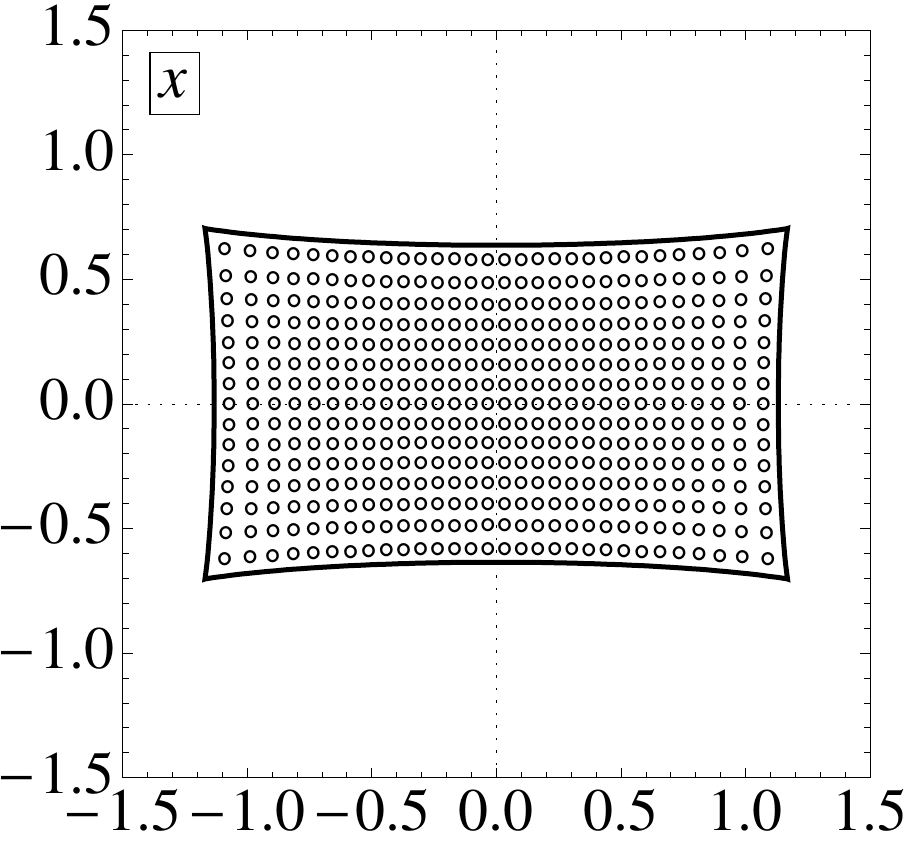}
\caption{\emph{The zeros of $H_{m,n}(m^{1/2}x)$ in the complex $x$-plane
for $(m,n,r)=\left(50,5,10\right)$ (left), 
$(m,n,r)=\left(40,10,4\right)$ (center), and 
$(m,n,r)=\left(30,15,2\right)$ (right), along 
with the boundary of the elliptic region $E_r$ that depends only on $r=m/n$.}}
\label{Hmn-zeros1}
\end{figure}
\begin{figure}[h]
\includegraphics[width=2.1in]{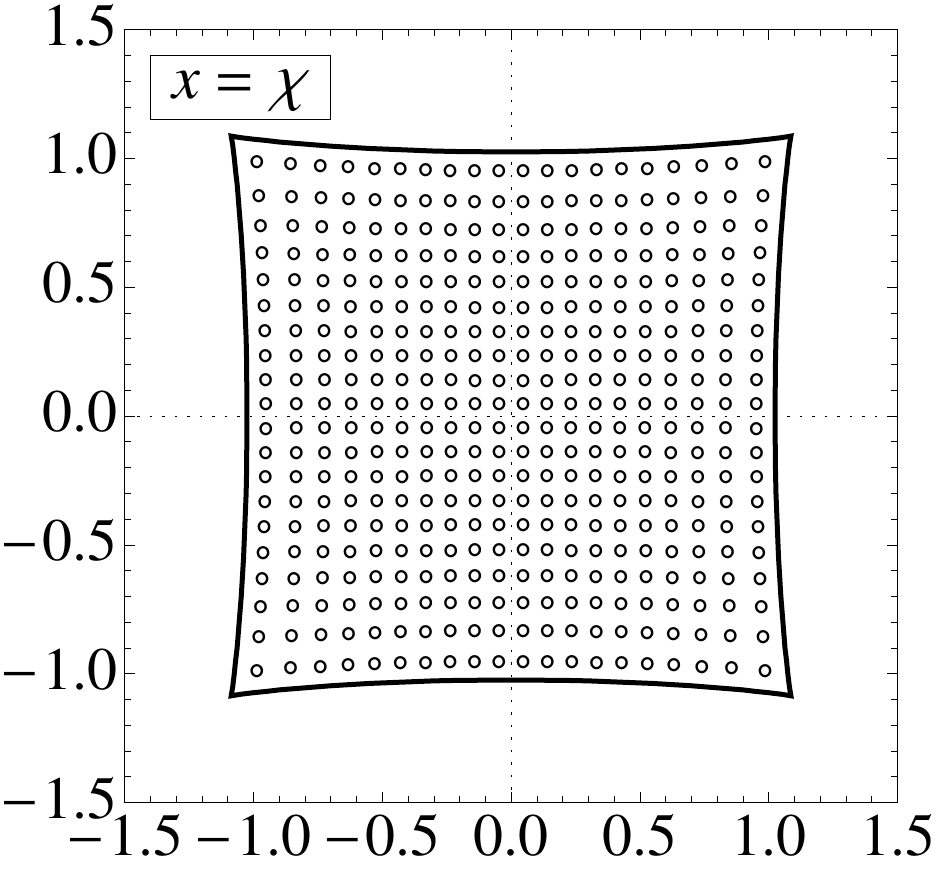}
\caption{\emph{The zeros of $H_{m,n}(m^{1/2}x)=H_{m,n}(n^{1/2}\chi)$ in the 
complex $x$-plane for $(m,n,r)=\left(20,20,1\right)$, along with the boundary 
of the elliptic region $E_r$ that depends only on $r=m/n$.}}
\label{Hmn-zeros2}
\end{figure}
\begin{figure}[h]
\includegraphics[width=2.1in]{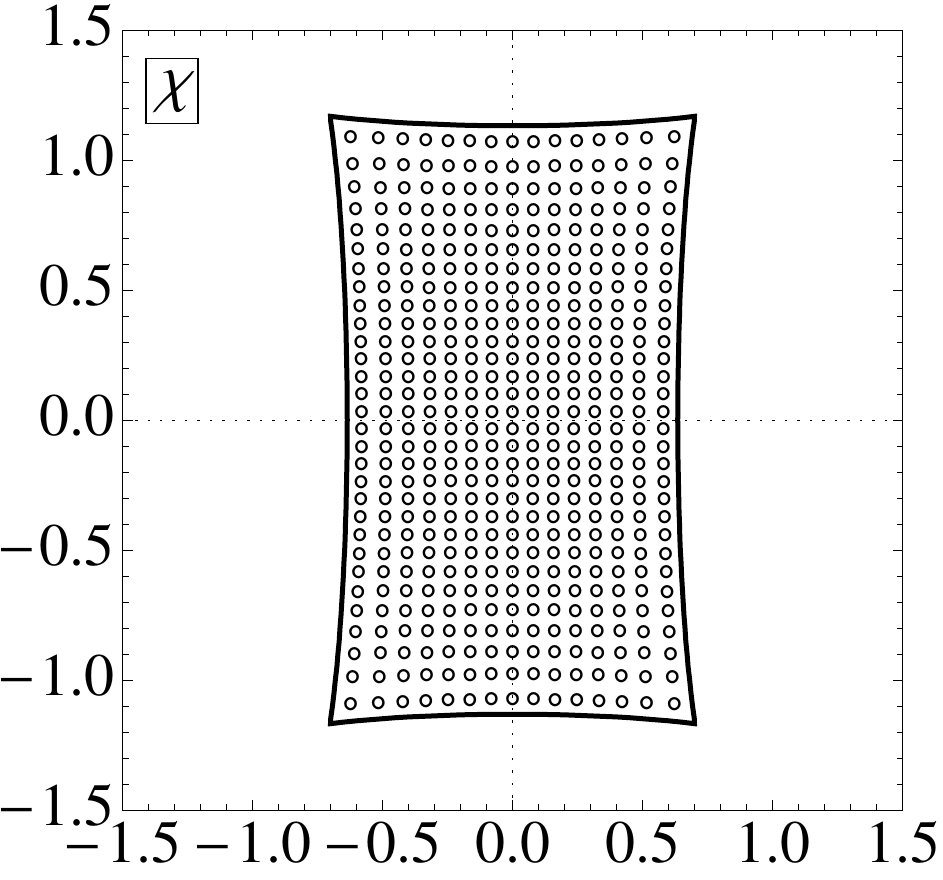}
\includegraphics[width=2.1in]{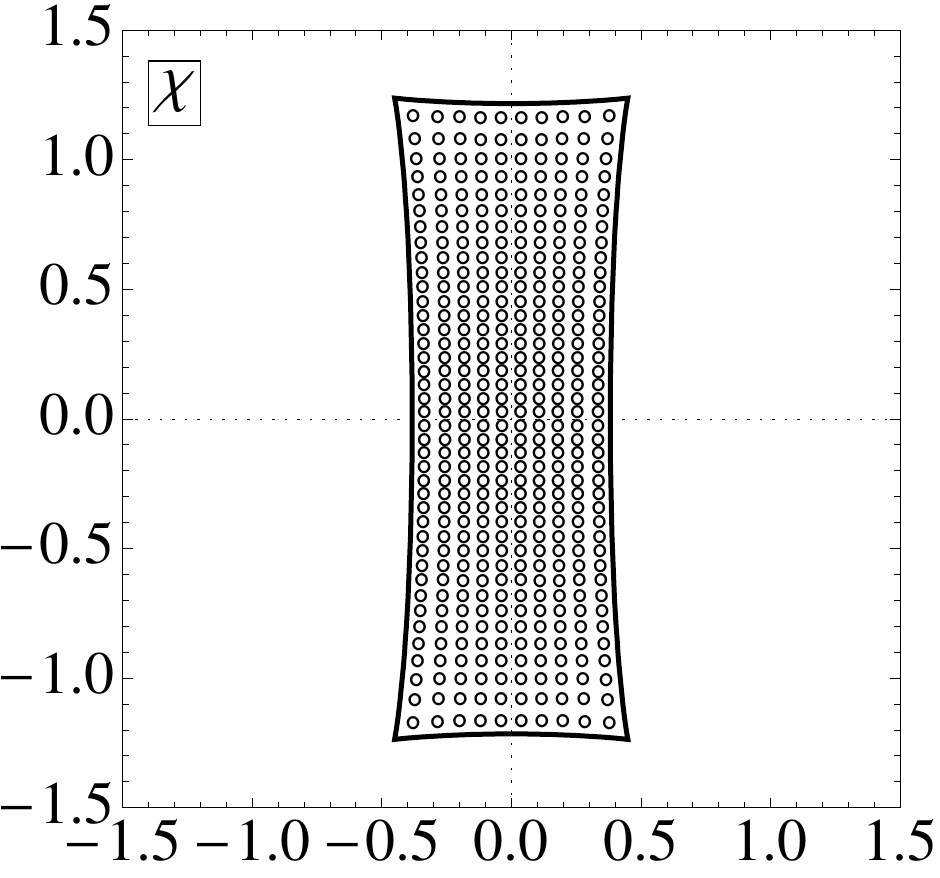}
\includegraphics[width=2.1in]{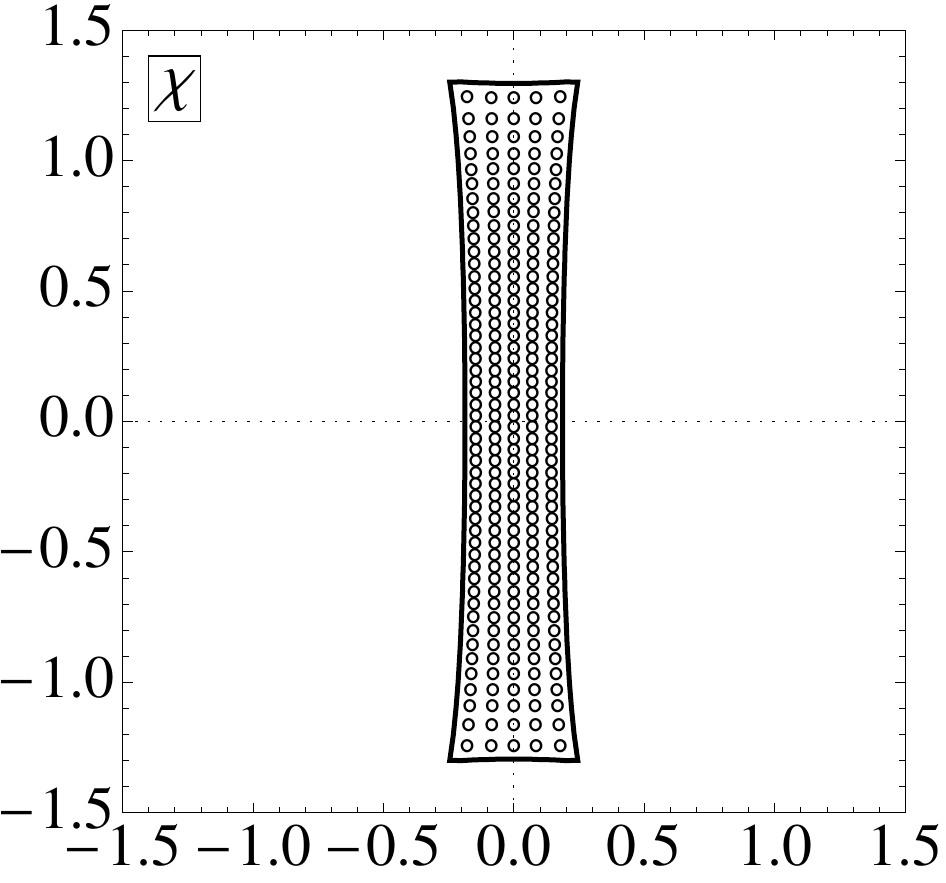}
\caption{\emph{The zeros of $H_{m,n}(n^{1/2}\chi)$ in the complex $\chi$-plane
for 
$(m,n,r)=\left(15,30,\tfrac{1}{2}\right)$ (left), 
$(m,n,r)=\left(10,40,\tfrac{1}{4}\right)$ (center), and 
$(m,n,r)=\left(5,50,\tfrac{1}{10}\right)$ (right), 
along 
with the boundary of the elliptic region $E_r$ that depends only on $r=m/n$.}}
\label{Hmn-zeros3}
\end{figure}

Various other geometric patterns are also seen in the 
plots of poles and zeros of rational solutions of the Painlev\'e-II equation 
and equations in the Painlev\'e-II hierarchy \cite{Clarkson-Mansfield:2003}, 
the Painlev\'e-III equation \cite{Clarkson:2003b}, systems of the symmetric 
Painlev\'e-IV hierarchy \cite{Filipuk:2008}, and the Painlev\'e-V equation 
\cite{Clarkson:2005}, as well as certain Wronskians of Hermite polynomials 
that are extensions of the generalized Hermite polynomials and have connections 
to Young diagrams \cite{FelderHV:2012}.  
Recently, significant progress has been made in understanding the rational 
solutions of the Painlev\'e-II equation, which can be indexed by a single 
integer $m$.  As $m\to\infty$, appropriately scaled zeros and poles of 
these rational functions densely fill a region $T$ bounded by a curvilinear 
triangle.  By analyzing a Riemann-Hilbert problem derived from the 
Garnier-Jimbo-Miwa Lax pair, the large-$m$ behavior of these functions (and 
certain functions arising in the study of critical behavior in the 
semiclassical sine-Gordon equation whose logarithmic derivatives are the 
rational Painlev\'e-II functions \cite{Buckingham:2012}) was rigorously 
calculated with error terms outside $T$ in terms of elementary functions, 
inside $T$ in terms of Riemann theta functions, along edges of $T$ in terms of 
trigonometric functions, and at corners of $T$ in terms of the tritronqu\'ee 
Painlev\'e-I solution \cite{Buckingham:2013,Buckingham:2014}.  In a later 
work, Bertola and Bothner \cite{Bertola:2014} reproduced part of these 
results, in particular the equation for the boundary of $T$ and information 
about the location of the zeros and poles, by deriving a new determinantal 
formula for the squares of the associated Yablonskii-Vorob'ev polynomials 
and applying Riemann-Hilbert analysis to a related family of orthogonal 
polynomials.  Joint with Balogh, they also used their method to obtain the 
boundary of the zero region for the generalized Yablonskii-Vorob'ev polynomials 
associated to the Painlev\'e-II hierarchy \cite{BaloghBB:2016}.
Miller and Sheng \cite{MillerS:2017} have recently shown that, 
for monodromy data corresponding to rational solutions, the Riemann-Hilbert 
problem associated to the Flaschka-Newell Painlev\'e-II Lax pair is equivalent 
to the Riemann-Hilbert problem for orthogonal polynomials studied by 
Bertola and Bothner.

In this work we use the Bertola-Bothner orthogonal polynomial approach to 
analyze the rational Painlev\'e-IV functions associated to the generalized 
Hermite polynomials.  Set
\eq
\begin{matrix}[l]
\alpha_{m,n}^{(I)}:=2m+n+1, & \beta_{m,n}^{(I)}:=-2n^2, & \mathscr{P}_{-1/z}^{(I)} := \{(\alpha_{m,n}^{(I)},\beta_{m,n}^{(I)}):m\geq 0,n\geq 1\}, \\
\alpha_{m,n}^{(II)}:=-(m+2n+1), & \beta_{m,n}^{(II)}:=-2m^2, & \mathscr{P}_{-1/z}^{(II)} := \{(\alpha_{m,n}^{(II)},\beta_{m,n}^{(II)}):m\geq 1,n\geq 0\}, \\
\alpha_{m,n}^{(III)}:=n-m, & \beta_{m,n}^{(III)}:=-2(m+n+1)^2, & \mathscr{P}_{-2z}^{(III)} := \{(\alpha_{m,n}^{(III)},\beta_{m,n}^{(III)}):m,n\in\mathbb{N}_0\},\\
\alpha_{j,k}^\text{(Oka)}:=j, & \beta_{j,k}^\text{(Oka)}:=-2(2k-j+\tfrac{1}{3})^2, & \mathscr{P}_{-2/(3z)}^\text{(Oka)}:=\{(\alpha_{j,k}^\text{(Oka)},\beta_{j,k}^{(Oka)}):j,k\in\mathbb{Z}\}.
\end{matrix}
\endeq
where $\mathbb{N}_0$ denotes the nonnegative integers.  It is known that 
the Painlev\'e-IV equation \eqref{p4} has a rational solution if and only if 
$(\alpha,\beta)\in\mathscr{P}_{-1/z}^{(I)}\cup\mathscr{P}_{-1/z}^{(II)}\cup\mathscr{P}_{-2z}^{(III)}\cup\mathscr{P}_{-2/(3z)}^\text{(Oka)}$.  
Furthermore, for fixed $(\alpha,\beta)$ this rational solution is unique 
when it exists \cite{Murata:1985,Kajiwara:1998,Noumi:1999}.  The 
families of rational solutions to \eqref{p4} 
corresponding to $\mathscr{P}_{-1/z}^{(I)}\cup\mathscr{P}_{-1/z}^{(II)}$, 
$\mathscr{P}_{-2z}^{(III)}$, and $\mathscr{P}_{-2/(3z)}^\text{(Oka)}$ are 
referred to as the $-1/z$, $-2z$, and $-2/(3z)$ hierarchies, respectively.  
The rational functions corresponding to $\mathscr{P}_{-2/(3z)}^\text{(Oka)}$ 
can be constructed from the generalized Okamoto polynomials.  The rational 
solutions of \eqref{p4} for 
$(\alpha,\beta)\in\mathscr{P}_{-1/z}^{(I)}\cup\mathscr{P}_{-1/z}^{(II)}\cup\mathscr{P}_{-2z}^{(III)}$ can be contructed 
from generalized Hermite polynomials.  We will analyze these rational solutions 
in the remainder of this work.

The generalized Hermite polynomials $H_{m,n}(y)$ are defined for $m,n\in\mathbb{N}_0$ 
by the recurrence relations 
\eq
\begin{split}
2mH_{m+1,n}H_{m-1,n} & = H_{m,n}H_{m,n}'' - (H_{m,n}')^2 + 2mH_{m,n}^2, \\
2nH_{m,n+1}H_{m,n-1} & = -H_{m,n}H_{m,n}'' + (H_{m,n}')^2 + 2nH_{m,n}^2
\end{split}
\endeq
and the initial conditions
\eq
H_{0,0}=H_{1,0}=H_{0,1}=1, \quad H_{1,1}=2y.
\endeq 
The name arises from the fact that 
\eq
H_{m,1}(y) = H_m(y) \quad \text{and} \quad H_{1,n}(y) = i^{-n}H_n(iy),
\endeq
where for $m\in\mathbb{N}_0$, $H_m(y)$ is the standard Hermite polynomial 
defined by the generating function 
\eq
e^{2sy-s^2} = \sum_{n=0}^\infty\frac{H_n(y)s^n}{n!}.
\label{Hermite-generating-fun}
\endeq
The generalized Hermite polynomials also have the symmetry 
\eq
H_{m,n}(iy) = i^{mn}H_{n,m}(y).
\label{Hmn-symmetry}
\endeq
While we will not use them, it is interesting to note that their zeros satisfy 
various sum relations that generalize the Stieltjes relations for the zeros 
of Hermite polynomials \cite{KudryashovD:2007}.
The connection to the rational Painlev\'e-IV functions is that 
\eq
w_{m,n}^{(I)}(y) := \frac{\text{d}}{\text{d}y}\log\left(\frac{H_{m+1,n}(y)}{H_{m,n}(y)}\right)
\label{w-mn-I}
\endeq
solves the Painlev\'e-IV equation \eqref{p4} with parameters 
$(\alpha,\beta)=(\alpha_{m,n}^{(I)},\beta_{m,n}^{(I)})$,
\eq
w_{m,n}^{(II)}(y) := -\frac{\text{d}}{\text{d}y}\log\left(\frac{H_{m,n+1}(y)}{H_{m,n}(y)}\right)
\label{w-mn-II}
\endeq
solves \eqref{p4} with parameters 
$(\alpha,\beta)=(\alpha_{m,n}^{(II)},\beta_{m,n}^{(II)})$, and
\eq
w_{m,n}^{(III)}(y) := -2y+\frac{\text{d}}{\text{d}y}\log\left(\frac{H_{m,n+1}(y)}{H_{m+1,n}(y)}\right) = -2y - w_{m,n}^{(I)}(y) - w_{m,n}^{(II)}(y)
\label{w-mn-III}
\endeq
solves \eqref{p4} for 
$(\alpha,\beta)=(\alpha_{m,n}^{(III)},\beta_{m,n}^{(III)})$.  

\subsection{Outline and results}
Our starting point is the known identity \eqref{taumn-ito-Hmn} expressing the 
generalized Hermite polynomial $H_{m,n}$ in terms of a Hankel determinant of 
Hermite polynomials.  In Lemma \ref{lemma-switch-dets} we rewrite this as a 
Hankel determinant of certain moments (defined in \eqref{mu-def}) of a measure 
supported on the unit circle.  This establishes a connection to the associated 
orthogonal polynomials on the unit circle (see \eqref{orth-poly-def}), and we 
write the rational Painlev\'e-IV functions in terms of these orthogonal 
polynomials and their normalization constants in \eqref{wI-ito-psi} and  
\eqref{wII-ito-psi} (see also \eqref{w-mn-III}).  We write down the standard 
Riemann-Hilbert problem associated to the orthogonal polynomials, and show 
how to directly extract the rational Painlev\'e-IV functions from the 
Riemann-Hilbert problem in Lemmas \ref{wI-ito-N-lemma} and 
\ref{wII-ito-N-lemma}.  

In \S\ref{sec-boundary} we compute the so-called $g$-function, a 
standard tool used to regularize the Riemann-Hilbert problem and turn 
oscillatory jumps into constants.  By studying topological changes in the level 
lines of the related phase function $\varphi$, we derive an explicit form of 
the boundary curve, which we now state.  Fix $r\in[1,\infty)$.  Let $x_c(r)$ be 
the unique value of $x$ satisfying 
\eq
\label{xc-polynomial}
r^4x^8-24r^2(r^2+r+1)x^4+32r(2r^3+3r^2-3r-2)x^2 - 48(r^2+r+1)^2 = 0
\endeq
with $\Re(x_c)>0$ and $\Im(x_c)>0$.  The four points 
$\{\pm x_c,\pm\overline{x_c}\}$ will be the four corners of the boundary of the 
elliptic region as well as the four branch points of a function $Q$ we will 
define shortly.  While it is 
possible to solve \eqref{xc-polynomial} exactly since it is a quartic in 
$x^2$, we simply note that for $r=1$ the corner points are the four $x$ 
values satisfying 
\eq
x^4 = 36-24\sqrt{3} \quad (r=1),
\endeq
so that $x_c(1)\approx 1.086+1.086i$ (compare Figure \ref{Hmn-zeros2}).  Now 
define $Q(x;r)$ as the unique function satisfying 
\eq
\label{gen0-Q-quartic}
3(1+r)^2Q^4+8(1+r)r^{1/2}xQ^3+4(r-1+rx^2)Q^2-4=0
\endeq
such that $Q(x;r)=-x+\mathcal{O}(x^{-2})$ as $x\to+\infty$ and cut as shown 
in Figure \ref{fig-Q-branches}.  
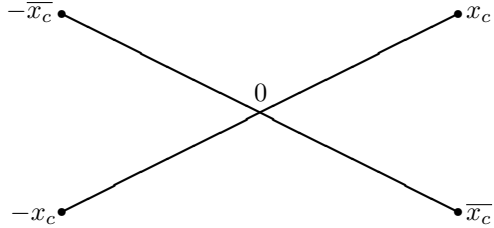
\begin{figure}[h]
\setlength{\unitlength}{1.5pt}
\begin{center}
\begin{picture}(100,50)(-50,-25)
\thicklines
\put(-50,-25){\line(2,1){100}}
\put(-50,25){\line(2,-1){100}}
\put(50,25){\circle*{2}}
\put(52,24){$x_c$}
\put(50,-25){\circle*{2}}
\put(52,-27){$\overline{x_c}$}
\put(-50,25){\circle*{2}}
\put(-64,24){$-\overline{x_c}$}
\put(-50,-25){\circle*{2}}
\put(-63,-27){$-x_c$}
\put(-1.5,3){$0$}
\end{picture}
\end{center}
\caption{\label{fig-Q-branches} The branch cuts for $Q(x;r)$.}
\end{figure}
Also define
\eq
S(x;r) := (1+r)Q(x;r)^3+2r^{1/2}xQ(x;r)^2.
\label{gen0-S-ito-Q}
\endeq
Then let $a(x;r)$ and $b(x;r)$ be the two values of $z$ satisfying 
\eq
\label{ab-def}
z^2-S(x;r)z+Q(x;r)^2 = 0.
\endeq
For definiteness we choose $\Im(a)<\Im(b)$ for 
$\arg(\overline{x_c})\leq \arg(x) \leq \arg(x_c)$ and 
$\Re(a)>\Re(b)$ for $\arg(x_c)\leq \arg(x) \leq \arg(-\overline{x_c})$.  
Throughout we restrict our analysis to 
$\arg(\overline{x_c}) \leq \arg(x) \leq \arg(-\overline{x_c})$, which is 
sufficient due to the symmetry \eqref{Hmn-symmetry}.  
We now specify a contour $\Sigma$ connecting $a$ and $b$.  
Define
\eq
\label{Rtilde-def}
\widetilde{R}(z;x,r):=(z^2-S(x;r)z+Q(x;r)^2)^{1/2}
\endeq
with $\widetilde{R}(z)=z+\mathcal{O}(1)$ as $z\to\infty$ and branch cut 
chosen as the straight line segment between $a$ and $b$.  Now define 
$\widetilde{\varphi}(z;x,r)\equiv\widetilde{\varphi}(z)$ by  
\eq
\label{phitilde-def}
\begin{split}
\widetilde{\varphi}(z) := & \frac{\widetilde{R}(z)}{Qz^2} + \left(1+r-\frac{S}{2Q^3}\right)\frac{\widetilde{R}(z)}{z} - (1+r)\log(2z+2\widetilde{R}(z)-S) \\ 
  & + (r-1)\log\left(\frac{2Q\widetilde{R}(z)-Sz+2Q^2}{z}\right) + \log(S^2-4Q^2) - (1+r)i\pi.
\end{split}
\endeq
Here all logarithms are chosen with principal branches (as we will only need 
the real part of $\widetilde{\varphi}$ the particular choice is unimportant).  
There is a level 
line of $\Re(\widetilde{\varphi}(z))$ connecting $a$ and $b$ traveling in the 
clockwise direction around the origin;  we call this bounded contour $\Sigma$.  
Now set $R(z;x,r)$ to be the function satisfying 
\eq
\label{R-def}
R(z;x,r)^2 = z^2-S(x;r)z+Q(x;r)^2
\endeq
that is analytic for $z\notin\Sigma$ and satisfies $R(z)=z+\mathcal{O}(1)$ as 
$z\to\infty$.  
Note we have the useful relations 
\eq
S=a+b, \quad Q=R(0), \quad Q^2=ab.
\endeq
Also define 
\eq
R_c(x;r) \equiv R_c := -\frac{((1+r)^2Q^4+2(1+r)QS+4)^{1/2}}{(1+r)Q},
\endeq
with the choice of branch inherited from $R(z)$.  Then we have the following 
definitions of the elliptic region in which the zeros and poles of the 
rational Painlev\'e-IV functions lie (at least asymptotically) and the 
complementary genus-zero region.  See Figures \ref{Hmn-zeros1} and 
\ref{Hmn-zeros2}.  
\begin{definition}
\label{genus-zero-def}
Fix $r\in[1,\infty)$.  Then the \emph{elliptic region} $E_r$ is the bounded 
domain of the complex plane defined by the curves
\eq
\label{boundary-curve}
\begin{split}
\Re\bigg\{ & \frac{(1+r)r^{1/2}x}{2} R_c - (1+r)\log\left(2R_c-\frac{4}{(1+r)Q}-S\right) \\ 
  & + (r-1)\log\left((1+r)Q^3+(1+r)Q^2R_c+S\right) + \log\left(S^2-4Q^2\right) \bigg\} = 0.
\end{split}
\endeq
The \emph{genus-zero region} is the complement of the closure of the elliptic 
region.  
\end{definition}

In \S\ref{sec-rhp-analysis} we carry out the Deift-Zhou nonlinear 
steepest-descent analysis \cite{DeiftZ:1993}  of the Riemann-Hilbert problem 
for the orthogonal polynomials.   This consists of several standard steps:
\begin{enumerate}
\item Conjugating the jump matrices by a matrix involving the $g$-function, 
which identifies the contours that will contribute to the leading-order 
solution.
\item Opening \emph{lenses} so all jumps are constants or decaying to 
the identity as $n\to\infty$.  
\item Solving the \emph{model problem} obtained by disregarding jumps close 
to the identity.
\item Controlling the errors and showing that the model solution is a good 
approximation to the exact problem.
\end{enumerate}
Following this procedure, we obtain the following asymptotic formulas for the 
rational Painlev\'e-IV functions valid in the genus-zero region.
\begin{theorem}
\label{thm-wI}
Fix $p,q\in\mathbb{N}$ (the positive integers) with $p\geq q$ and set 
$r:=p/q$.  Fix $x$ in the genus-zero region as defined in Definition 
\ref{genus-zero-def}.  Then as $m,n\to\infty$ along the sequence 
$\{m,n\}=\{jp,jq\}$ for $j\in\mathbb{N}$ we have 
\eq
\label{wI-ito-Q-and-S}
\frac{1}{n^{1/2}}w_{m,n}^{(I)}(m^{1/2}x) = -\frac{1}{Q(x,r)} - \frac{S(x,r)}{2Q(x,r)^2} + \mathcal{O}\left(\frac{1}{n}\right).
\endeq
\end{theorem}
Theorem \ref{thm-wI} is illustrated in Figure \ref{wI-approx-plots}.  
\begin{figure}[h]
\includegraphics[width=2.1in]{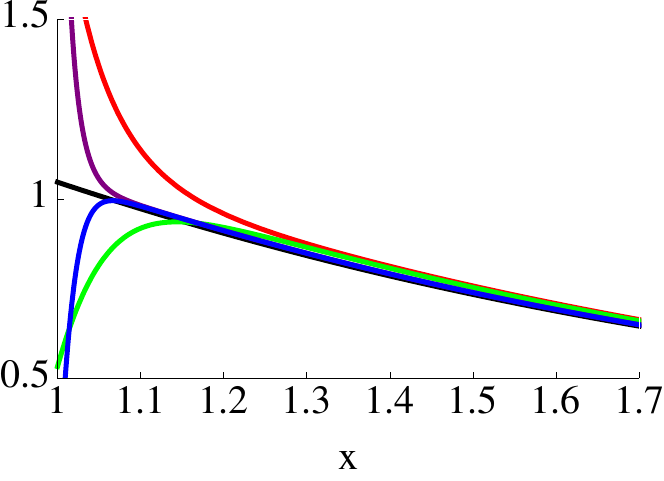}
\includegraphics[width=2.1in]{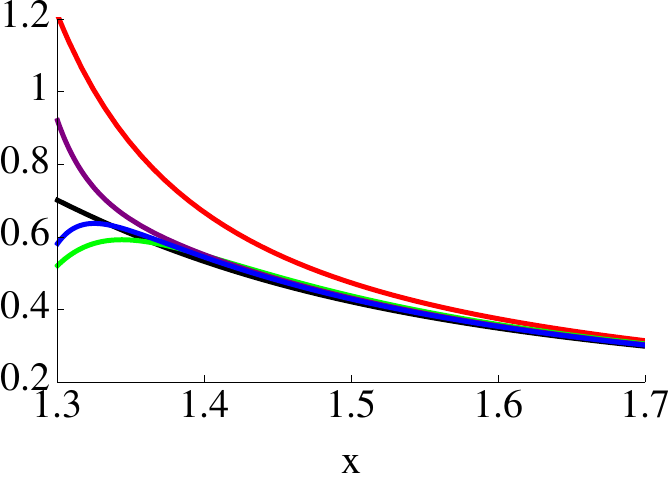}
\caption{Comparison of the rational Painlev\'e-IV functions of the first type 
with the limiting genus-zero approximation on the positive real $x$-axis 
outside the root region for $r=1$ (left) and $r=10$ (right).  {\it Left:}  
$-Q(x,1)^{-1}-\frac{1}{2}S(x,1)Q(x,1)^{-2}$ (black) plotted against 
$n^{-1/2}w_{m,n}^{(I)}(m^{1/2}x)$ for $m=n=5$ (red), $m=n=6$ (green), 
$m=n=21$ (purple), and $m=n=22$ (blue).  For $r=1$ the boundary of the root 
region intersects the positive real axis at $x\approx 1.0253$.  {\it Right:}  
$-Q(x,10)^{-1}-\frac{1}{2}S(x,10)Q(x,10)^{-2}$ (black) plotted against 
$n^{-1/2}w_{m,n}^{(I)}(m^{1/2}x)$ for $m=10$, $n=1$ (red), $m=20$, $n=2$ 
(green), $m=30$, $n=3$ (purple), and $m=40$, $n=4$ (blue).  For $r=10$ the 
boundary of the root region intersects the positive real axis at 
$x\approx 1.2953$.}
\label{wI-approx-plots}
\end{figure}
\begin{theorem}
\label{thm-wII}
Fix $p,q\in\mathbb{N}$ with $p\geq q$ and set 
$r:=p/q$.  Fix $x$ in the genus-zero region as defined in Definition 
\ref{genus-zero-def}.  Then as $m,n\to\infty$ along the sequence 
$\{m,n\}=\{jp,jq\}$ for $j\in\mathbb{N}$ we have 
\eq
\label{wII-ito-Q-and-S}
\frac{1}{n^{1/2}}w_{m,n}^{(II)}(m^{1/2}x) = \frac{1}{Q(x,r)} - \frac{S(x,r)}{2Q(x,r)^2} + \mathcal{O}\left(\frac{1}{n}\right).
\endeq
\end{theorem}
Theorem \ref{thm-wII} is illustrated in Figure \ref{wII-approx-plots}.  
\begin{figure}[h]
\includegraphics[width=2.1in]{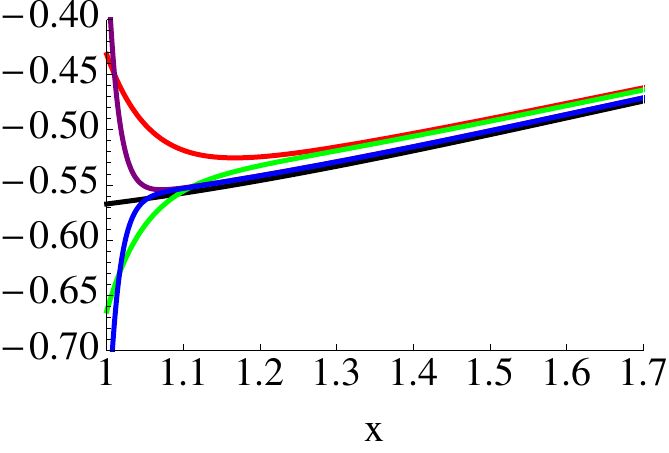}
\includegraphics[width=2.1in]{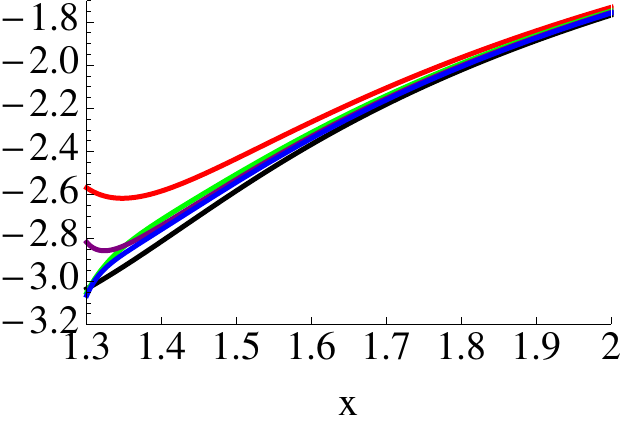}
\caption{Comparison of the rational Painlev\'e-IV functions of the second 
type with the limiting genus-zero approximation on the positive real $x$-axis 
outside the root region for $r=1$ (left) and $r=10$ (right).  {\it Left:}  
$Q(x,1)^{-1}-\frac{1}{2}S(x,1)Q(x,1)^{-2}$ (black) plotted against 
$n^{-1/2}w_{m,n}^{(II)}(m^{1/2}x)$ for $m=n=5$ (red), $m=n=6$ (green), 
$m=n=21$ (purple), and $m=n=22$ (blue).  For $r=1$ the boundary of the root 
region intersects the positive real axis at $x\approx 1.0253$.  {\it Right:}  
$Q(x,10)^{-1}-\frac{1}{2}S(x,10)Q(x,10)^{-2}$ (black) plotted against 
$n^{-1/2}w_{m,n}^{(II)}(m^{1/2}x)$ for $m=10$, $n=1$ (red), $m=20$, $n=2$ 
(green), $m=30$, $n=3$ (purple), and $m=40$, $n=4$ (blue).  For $r=10$ the 
boundary of the root region intersects the positive real axis at 
$x\approx 1.2953$.}
\label{wII-approx-plots}
\end{figure}
Finally, combining these two theorems with \eqref{w-mn-III} immediately gives 
the following. 
\begin{theorem}
\label{thm-wIII}
Fix $p,q\in\mathbb{N}$ with $p\geq q$ and set 
$r:=p/q$.  Fix $x$ in the genus-zero region as defined in Definition 
\ref{genus-zero-def}.  Then as $m,n\to\infty$ along the sequence 
$\{m,n\}=\{jp,jq\}$ for $j\in\mathbb{N}$ we have 
\eq
\frac{1}{n^{1/2}}w_{m,n}^{(III)}(m^{1/2}x) = -2r^{1/2}x + \frac{S(x,r)}{Q(x,r)^2} + \mathcal{O}\left(\frac{1}{n}\right).
\endeq
\end{theorem}

\begin{remark}
To understand the behavior of the rational Painlev\'e-IV functions as 
$m,n\to\infty$ with $r=m/n$ fixed, it is sufficient to consider the case 
$r\geq 1$ due to the symmetry \eqref{Hmn-symmetry}.  In the case $0<r<1$ (see 
Figure \ref{Hmn-zeros3}) the natural variable is $\chi:=n^{-1/2}y$, since the 
zeros of $H_{m,n}(n^{1/2}\chi)$ are bounded in the $\chi$ plane as 
$n\to\infty$.  
\end{remark}

\subsection{A comment on the literature}
Before beginning our analysis we make a few remarks regarding a recent paper 
by Novokshenov and Schelkonogov \cite{Novokshenov:2015} that 
concerns some of the same questions we address here.  In particular, they 
are interested in the distribution of the zeros of $w_{n,n}^{(III)}$ for 
large $n$.  The proposed strategy is intriguing:  determine a Riemann-Hilbert 
problem for $w_{0,0}^{(III)}$ and then apply Schlesinger/B\"acklund 
transformations to obtain Riemann-Hilbert problems for $w_{n,n}^{(III)}$.  
Unfortunately, \cite[Equation (22)]{Novokshenov:2015} expressing 
$w_{n,n}^{(III)}$ (or, in their notation, $u_{n,n}$) in terms of the solution 
of the Riemann-Hilbert problem in \cite[Equation (21)]{Novokshenov:2015} is not 
correct.  This means that the subsequent asymptotic results for the rational 
Painlev\'e-IV functions are also incorrect, including 
\cite[Equation (37)]{Novokshenov:2015} and 
\cite[Equation (38)]{Novokshenov:2015} describing the asymptotic behavior of 
$w_{n,n}^{(III)}$ and \cite[Equation (41)]{Novokshenov:2015} for the location 
of the zeros.  In fact, it is not possible to extract any information about 
$w_{n,n}^{(III)}$ from the Riemann-Hilbert problem in 
\cite[Equation (21)]{Novokshenov:2015}.  In their notation, this 
problem is to find a matrix $Y(\xi)$ analytic for $\xi\notin\mathbb{R}$ 
satisfying
\eq
Y_+(\xi)=Y_-(\xi)\bbm 1 & 2\pi i e^{-n(\xi^2-x^2)} \\ 0 & 1 \ebm \text{ for } \xi\in\mathbb{R};  \quad Y(\xi)=(\mathbb{I}+\mathcal{O}(\xi^{-1}))\bbm \xi^{2n} & 0 \\ 0 & \xi^{-2n} \ebm \text{ as }\xi\to\infty
\endeq
(here the parameter $x$ is, after scaling, the independent variable for the 
Painlev\'e-IV functions and is the same as our $x$ defined in 
\eqref{x-chi-z-scales} if $m=n$).  Then the function 
$(2\pi ie^{nx^2})^{-\sigma_3/2} Y(\xi) (2\pi ie^{nx^2})^{\sigma_3/2}$
satisfies a Fokas-Its-Kitaev Riemann-Hilbert problem \cite{Fokas:1991} 
for the (standard) Hermite polynomials.  The solution to this problem can be 
written exactly in terms of $H_n$, $H_{n-1}$, and their Cauchy transforms, 
which is not enough information to construct $H_{n,n}$ or $w_{n,n}^{(III)}$.

\subsection{Notation}
We denote the positive integers by $\mathbb{N}$ and the nonnegative integers 
by $\mathbb{N}_0$.  
If $f$ is a function defined on a specified oriented contour, then $f_+$ 
($f_-$) denotes the boundary value taken from the left (right).  Matrices are 
denoted by bold capital letters, with the exception of the $2\times 2$ 
identity matrix $\mathbb{I}$ and the Pauli matrix
\eq
\sigma_3:=\bbm 1 & 0 \\ 0 & -1 \ebm.
\endeq
The $(jk)$-entry of a matrix ${\bf M}$ is denoted by $[{\bf M}]_{jk}$.

\noindent
{\bf Acknowledgements.}  The author thanks Ferenc Balogh, Thomas Bothner, 
Walter Van Assche, Peter Miller, and Arno Kuijlaars for helpful discussions, 
the Charles Phelps Taft Research Center for a Faculty Release Fellowship, and 
the National Science Foundation for support via grants DMS-1312458 and 
DMS-1615718.

\section{The associated orthogonal polynomials}

To analyze the asymptotic behavior of these functions we will use a 
determinantal formula.  
Define $\tau_{m,n}(y)$ by $\tau_{m,0}(y):=1$ and by the $n\times n$ Hankel determinant 
\eq
\tau_{m,n}(y) := \bvm H_m(y) & H_{m+1}(y) & \cdots & H_{m+n-1}(y) \\ H_{m+1}(y) & H_{m+2}(y) & \cdots & H_{m+n}(y) \\ \vdots & \vdots & \ddots & \vdots \\ H_{m+n-1}(y) & H_{m+n}(y) & \cdots & H_{m+2n-2}(y) \evm_{n\times n}
\endeq
for $n\geq 1$.  Then $\tau_{m,n}$ is related \cite{Kajiwara:1998,Noumi:1999}
to the generalized Hermite polynomial $H_{m,n}$ by 
\eq
\tau_{m,n}(y) = (-1)^{\lceil(n-1)/2\rceil} \left(\prod_{k=0}^{n-1}[k!2^k]\right) H_{m,n}(y),
\label{taumn-ito-Hmn}
\endeq
where $\lceil\cdot\rceil$ denotes the ceiling function.  We rewrite 
$\tau_{m,n}$ in terms of certain moments as follows.  
Let the contour $C$ be the unit circle with clockwise 
orientation.  For $\zeta\in C$, define the measure 
\eq
\d\nu_m(\zeta;y):=\exp\left(\frac{2y}{\zeta}-\frac{1}{\zeta^2}\right)\zeta^m\frac{\text{d}\zeta}{2\pi i\zeta}.
\endeq
Define the moments 
\eq
\label{mu-def}
\mu_k^{(m)}(y) := -\oint_C\zeta^k\d\nu_m(\zeta;y).
\endeq
Now, via the generating function \eqref{Hermite-generating-fun}, the 
Cauchy integral formula for derivatives, and the change of variables 
$s=\zeta^{-1}$, we see we can write the standard Hermite polynomials as 
\eq
\begin{split}
H_{m+j}(y) & = \frac{\d^{m+j}}{\d^{m+j}s}\left.\left(e^{2sy-s^2}\right)\right\vert_{s=0} = -\frac{(m+j)!}{2\pi i}\oint_C \frac{e^{2sy-s^2}\d s}{s^{m+j+1}} \\ 
 & = -(m+j)!\oint_C\zeta^j\d\nu_m(\zeta;y) = (m+j)!\mu_j^{(m)}(y).
\end{split}
\endeq
In particular, this means we can write
\eq
\tau_{m,n}(y) = \bvm m!\mu_0^{(m)}(y) & (m+1)!\mu_1^{(m)}(y) & \cdots & (m+n-1)!\mu_{n-1}^{(m)}(y) \\ (m+1)!\mu_1^{(m)}(y) & (m+2)!\mu_2^{(m)}(y) & \cdots & (m+n)!\mu_n^{(m)}(y) \\ \vdots & \vdots & \ddots & \vdots \\ (m+n-1)!\mu_{n-1}^{(m)}(y) & (m+n)!\mu_n^{(m)}(y) & \cdots & (m+2n-2)!\mu_{2n-2}^{(m)}(y) \evm_{n\times n}.
\endeq
Define the related $n\times n$ Hankel determinant 
\eq
T_{m,n}(y) := \left|\mu_{j+k-2}^{(m)}(y)\right|_{j,k=1}^n = \bvm \mu_0^{(m)}(y) & \mu_1^{(m)}(y) & \cdots & \mu_{n-1}^{(m)}(y) \\ \mu_1^{(m)}(y) & \mu_2^{(m)}(y) & \cdots & \mu_n^{(m)}(y) \\ \vdots & \vdots & \ddots & \vdots \\ \mu_{n-1}^{(m)}(y) & \mu_n^{(m)}(y) & \cdots & \mu_{2n-2}^{(m)}(y) \evm_{n\times n}.
\endeq
Certain ratios of these determinants can be expressed in terms of normalization 
constants for a family of orthogonal polynomials (see \eqref{hmn-ito-Tmn} 
below).  We now show how to relate $\tau_{m,n}$ with $T_{m,n}$ (with shifted 
indices), thus providing a bridge between the rational Painlev\'e-IV functions 
and the orthogonal polynomials.  
\begin{lemma}
\eq
\tau_{m,n}(y) = \left(\prod_{k=0}^{n-1}\left[(m+k)!2^k\right]\right)\cdot T_{m-n+1,n}(y).
\label{taumn-ito-Tmn}
\endeq
\label{lemma-switch-dets}
\end{lemma}
\begin{proof}
We start by writing the right-hand side of \eqref{taumn-ito-Tmn} in terms of 
Hermite polynomials:
\eq
\label{det-identity-goal}
\left(\prod_{k=0}^{n-1}\left[(m+k)!2^k\right]\right)\cdot T_{m-n+1,n} = \prod_{k=1}^{n-1}2^k \bvm \frac{m!}{(m-n+1)!}H_{m-n+1} & \frac{(m+1)!}{(m-n+2)!}H_{m-n+2}& \hspace{-.1in} \cdots & \hspace{-.1in} \frac{(m+n-1)!}{m!}H_m \\ 
\frac{m!}{(m-n+2)!}H_{m-n+2} & \frac{(m+1)!}{(m-n+3)!}H_{m-n+3} & \hspace{-.1in} \cdots & \hspace{-.1in} \frac{(m+n-1)!}{(m+1)!}H_{m+1} \\
\vdots & \vdots & \hspace{-.1in} \ddots & \hspace{-.1in} \vdots \\ 
m H_{m-1} & (m+1)H_m & \hspace{-.1in} \cdots & \hspace{-.1in}(m+n-1)H_{m+n-2} \\
H_m & H_{m+1} & \hspace{-.1in} \cdots & \hspace{-.1in} H_{m+n-1}
\evm.
\endeq
Our goal is to manipulate $\tau_{m,n}$ into this form.  We start by 
completely reversing the order of the rows:
\eq
\tau_{m,n} = \prod_{k=1}^{n-1}(-1)^k\bvm 
H_{m+n-1} & H_{m+n} & \cdots & H_{m+2n-2} \\
H_{m+n-2} & H_{m+n-1} & \cdots & H_{m+2n-3} \\
H_{m+n-3} & H_{m+n-2} & \cdots & H_{m+2n-4} \\
\vdots & \vdots & \ddots & \vdots \\ 
H_{m+1} & H_{m+2} & \cdots & H_{m+n} \\ 
H_m & H_{m+1} & \cdots & H_{m+n-1}
\evm.
\endeq
Note that the $n$th row is in the desired form (up to the overall constant).  
We now perform a set of operations on the first $n-1$ rows that will leave 
the $(n-1)$st row in the desired form.  Repeating this set of operations on 
the first $n-2$ rows, then the first $n-3$ rows, and so on, will establish 
the identity.  The Hermite polynomials satisfy the recursion relation
\eq
H_{m+1}(y) = 2yH_m(y) - 2mH_{m+1}(y).
\endeq
Using this in the top row gives
\eq
\tau_{m,n} = \prod_{k=1}^{n-1}(-1)^k\bvm 
2yH_{m+n-2}-2(m+n-2)H_{m+n-3} & \cdots & 2yH_{m+2n-3} - 2(m+2n-3)H_{m+2n-4} \\
H_{m+n-2} & \cdots & H_{m+2n-3} \\
H_{m+n-3} & \cdots & H_{m+2n-4} \\
\vdots & \ddots & \vdots \\ 
H_{m+1} & \cdots & H_{m+n} \\ 
H_m & \cdots & H_{m+n-1}
\evm.
\endeq
Note that we can eliminate the terms proportional to $y$ by subtracting a 
multiple of the second row from the first row.  We can then pull out the 
common $-2$ factor from the first row, and subtract a multiple of the third 
row from the first row to change the coefficients in front of the Hermite 
polynomials.  The result is 
\eq
\tau_{m,n} = -2\prod_{k=1}^{n-1}(-1)^k\bvm 
m H_{m+n-3} & \cdots & (m+n-1)H_{m+2n-4} \\
H_{m+n-2} & \cdots & H_{m+2n-3} \\
H_{m+n-3} & \cdots & H_{m+2n-4} \\
\vdots & \ddots & \vdots \\ 
H_{m+1} & \cdots & H_{m+n} \\ 
H_m & \cdots & H_{m+n-1}
\evm.
\endeq
We now carry out the same procedure on rows $2,3,\dots,n-1$:  apply the 
recursion relation, use the next row to remove terms proportional to $y$, 
and then use the subsequent row to change the coefficient of the first 
entry to $m$.  (For row $n-1$ the leading coefficient in column 1 is already 
$m$ once the $y$-terms are removed).  Once every row has been modified in 
this way we obtain
\eq
\tau_{m,n} = 2^{n-1}\prod_{k=1}^{n-2}(-1)^k\bvm 
m H_{m+n-3} & \cdots & (m+n-1)H_{m+2n-4} \\
m H_{m+n-4} & \cdots & (m+n-1)H_{m+2n-5} \\
m H_{m+n-5} & \cdots & (m+n-1)H_{m+2n-6} \\
\vdots & \ddots & \vdots \\ 
m H_{m-1} & \cdots & (m+n-1)H_{m+n-2} \\ 
H_m & \cdots & H_{m+n-1}
\evm.
\endeq
This fixes the last two rows.  We now repeat this procedure on rows 
$1,...,n-2$, the only difference being that we change the leading 
coefficients in column 1 to $m(m-1)$.  The result is 
\eq
\tau_{m,n} = \prod_{k=n-2}^{n-1}2^k\prod_{j=1}^{n-3}(-1)^j\bvm 
m(m-1) H_{m+n-5} & \cdots & (m+n-1)(m+n-2)H_{m+2n-6} \\
m(m-1) H_{m+n-6} & \cdots & (m+n-1)(m+n-2)H_{m+2n-7} \\
\vdots & \ddots & \vdots \\ 
m(m-1) H_{m-2} & \cdots & (m+n-1)(m+n-2)H_{m+n-3} \\
m H_{m-1} & \cdots & (m+n-1)H_{m+n-2} \\ 
H_m & \cdots & H_{m+n-1}
\evm.
\endeq
Note that now the final three rows have the intended form.  Repeating this 
procedure $n-3$ more times, each time involving one less row than before and 
modifying the leading coefficient appropriately (i.e. so the last row 
changed has the correct coefficient), yields the form 
\eqref{det-identity-goal}, as desired.
\end{proof}
\begin{remark}
We observe that the result of Lemma \ref{lemma-switch-dets} can be written in 
terms of Hermite polynomials as
\eq
\left|H_{m+j+k-2}(y)\right|_{j,k=1}^n = \prod_{k=0}^{n-1}2^k\cdot\left|\frac{(m+k-1)!}{(m-n+j+k-1)!}H_{m-n+j+k-1}(y)\right|_{j,k=1}^{n}.
\endeq
Hankel determinants of orthogonal polynomials such as the expression on the 
left-hand side are known as \emph{Tur\'anians}.  The Hermite Tur\'anian can 
be expressed as a Wronskian for general $m$ 
\cite[Equation (18.2)]{KarlinS:1961} and evaluated in closed form for $m=0$ 
\cite[Equation (3.55)]{Krattenthaler:1999}.  For more background and 
references on Tur\'anians see \cite{Ismail:2005}.
\end{remark}
For fixed $m\in\mathbb{N}_0$, define the monic orthogonal polynomials 
$\psi_n^{(m)}$, $n\geq 0$, by 
\eq
\label{orth-poly-def}
\oint_C\psi_n^{(m)}(\zeta;y)\zeta^j\d\nu_m(\zeta;y) = \delta_{jn}h_n^{(m)}(y), \quad j=0,\dots,n,
\endeq
where $\delta_{jn}$ is the Kroneker delta function and $h_n^{(m)}(y)$ is the 
normalization constant (that is, constant in $\zeta$ but with parametric 
dependence on $y$).  
Then (see, for example, \cite{Bertola:2014,BaloghBB:2016}) the value of the orthogonal polynomials evaluated at $\zeta=0$ can be 
expressed in terms of determinants via 
\eq
\psi_n^{(m)}(0;y) = (-1)^n\frac{T_{m+1,n}(y)}{T_{m,n}(y)},
\label{psimn-ito-Tmn}
\endeq
and the normalization constant $h_n^{(m)}$ can be expressed as
\eq
h_n^{(m)}(y) = -\frac{T_{m,n+1}(y)}{T_{m,n}(y)}.
\label{hmn-ito-Tmn}
\endeq
Note that \eqref{psimn-ito-Tmn} and \eqref{hmn-ito-Tmn} provide ways to 
shift the two indices of $T_{m,n}(y)$.  Applying \eqref{taumn-ito-Hmn}, 
\eqref{taumn-ito-Tmn}, \eqref{psimn-ito-Tmn}, and \eqref{hmn-ito-Tmn} to 
\eqref{w-mn-I}--\eqref{w-mn-II} gives 
\eq
\label{wI-ito-psi}
w_{m,n}^{(I)}(y) = \frac{\text{d}}{\text{d}y}\log\left(\frac{\tau_{m+1,n}(y)}{\tau_{m,n}(y)}\right) = \frac{\text{d}}{\text{d}y}\log\left(\frac{T_{m-n+2,n}}{T_{m-n+1,n}}\right) = \frac{\partial}{\partial y}\log\left(\psi_n^{(m-n+1)}(0;y)\right)
\endeq
and
\eq
\label{wII-ito-psi}
w_{m,n}^{(II)}(y) = \frac{\text{d}}{\text{d}y}\log\left(\frac{\tau_{m,n}(y)}{\tau_{m,n+1}(y)}\right) = \frac{\text{d}}{\text{d}y}\log\left(\frac{T_{m-n+1,n}(y)}{T_{m-n,n+1}(y)}\right) = \frac{\partial}{\partial y}\log\left(\frac{\psi_n^{(m-n)}(0;y)}{h_n^{(m-n)}(y)}\right).
\endeq
Note that $w_{m,n}^{(III)}(y)$ can also be expressed in terms of the 
orthogonal polynomials and their normalization constants through the previous 
two equations and \eqref{w-mn-III}.  We now introduce the 
Fokas-Its-Kitaev Riemann-Hilbert problem \cite{Fokas:1991} in order to analyze 
the large-degree behavior of the orthogonal polynomials.  

\begin{rhp}[Unscaled orthogonal polynomial problem]
Fix $y\in\mathbb{C}$ and $m,n\in\mathbb{N}$.  
Seek a $2\times 2$ matrix ${\bf M}_{m,n}(\zeta;y)$ 
with the following properties:
\begin{itemize}
\item[]\textbf{Analyticity:}  ${\bf M}_{m,n}(\zeta;y)$
is analytic for $\zeta\in\mathbb{C}$ except on $C$ (the unit circle oriented 
clockwise) with H\"older-continuous boundary values.
\item[]\textbf{Jump condition:}  The boundary values taken by 
${\bf M}_{m,n}(\zeta;y)$ on $C$ are related by the jump condition 
\eq
\label{M-jump}
{\bf M}_{m,n+}(\zeta;y)={\bf M}_{m,n-}(\zeta;y) \bbm 1 & \displaystyle\frac{1}{2\pi i\zeta}\exp\left(\frac{2y}{\zeta}-\frac{1}{\zeta^2}+m\log\zeta\right) \\ 0 & 1 \ebm, \quad \zeta\in C.
\endeq 
\item[]\textbf{Normalization:}  As $\zeta\to\infty$, the matrix 
${\bf M}_{m,n}(\zeta;y)$ satisfies the condition
\begin{equation}
{\bf M}_{m,n}(\zeta;y) = (\mathbb{I}+\mathcal{O}(\zeta^{-1}))\zeta^{n\sigma_3}
\end{equation}
with the limit being uniform with respect to direction.
\end{itemize}
\label{rhp:M}
\end{rhp}
This Riemann-Hilbert problem is solvable exactly when $\psi_n^{(m)}$ exists, 
and 
\eq
\psi_n^{(m)}(\zeta;y) = [{\bf M}_{m,n}(\zeta;y)]_{11}
\endeq
(that is, the 11-entry of ${\bf M}$) while
\eq
h_n^{(m)}(y) = -2\pi i\lim_{\zeta\to\infty}\zeta\left[{\bf M}_{m,n}(\zeta;y)\zeta^{-n\sigma_3}-\mathbb{I}\right]_{12}.
\endeq
Motivated by the exponent in \eqref{M-jump}, we define rescaled versions of 
$y$ and $\zeta$:
\eq
\label{x-chi-z-scales}
x:=m^{-1/2}y, \quad z:=n^{1/2}\zeta.
\endeq
These definitions suggest scaling the orthogonal polynomials as well.  Define
\eq
\label{Psi-H-def}
\Psi_n^{(m)}(z;x):=n^{n/2}\psi_n^{(m)}\left(\frac{z}{n^{1/2}};m^{1/2}x\right), \quad \mathcal{H}_n^{(m)}(x):=n^{n+\frac{m}{2}}h_n^{(m)}(m^{1/2}x).
\endeq
These new polynomials satisfy the orthogonality relations 
\eq
\oint_C\Psi_n^{(m)}(z;x)z^jdV_m(z;x) = \delta_{jn}\mathcal{H}_n^{(m)}(x), \quad j=0,\dots,n, \quad dV_m:=\exp\left(n\left[\frac{2r^{1/2}x}{z}-\frac{1}{z^2}\right]\right)\frac{z^{r\cdot n}dz}{2\pi iz},
\endeq
where $r=m/n$.  

The desired rational functions can be expressed in terms of the scaled 
orthogonal polynomials as
\eq
\label{wI-ito-log-Psi}
m^{1/2}w_{m,n}^{(I)}(m^{1/2}x) = \frac{\partial}{\partial x}\log\left(\Psi_n^{(m-n+1)}(0;x)\right) = \frac{\frac{\partial}{\partial x}\Psi_n^{(m-n+1)}(0;x)}{\Psi_n^{(m-n+1)}(0;x)}
\endeq
and 
\eq
\label{wII-ito-log-Psi}
m^{1/2}w_{m,n}^{(II)}(m^{1/2}x) = \frac{\partial }{\partial x}\log\left(\frac{\Psi_n^{(m-n)}(0;x)}{\mathcal{H}_n^{(m-n)}(x)}\right).
\endeq

We now pose a Riemann-Hilbert problem for the orthogonal polynomials 
$\Psi_n^{(m-n+1)}(z;x)$.  
\begin{rhp}[Scaled orthogonal polynomial problem]
Fix $x\in\mathbb{C}$ and $m,n\in\mathbb{N}$ with $m\geq n$ and set 
$r=m/n$.  Find the unique $2\times 2$ matrix ${\bf N}_{m,n}(z;x)$ 
with the following properties:
\begin{itemize}
\item[]\textbf{Analyticity:}  ${\bf N}_{m,n}(z;x)$
is analytic in $z$ except on $C$ (the unit circle oriented clockwise)
with H\"older-continuous boundary values.
\item[]\textbf{Jump condition:}  The boundary values taken by 
${\bf N}_{m,n}(z;x)$ on $C$ are related by the jump condition 
\eq
\label{N-jump}
{\bf N}_{m,n+}(z;x)={\bf N}_{m,n-}(z;x) \bbm 1 & \displaystyle\frac{1}{2\pi i}e^{-n\theta(z;x,r)} \\ 0 & 1 \ebm, \quad z\in C,
\endeq 
where
\eq
\theta(z;x,r):=(1-r)\log z -\frac{2r^{1/2}x}{z} + \frac{1}{z^2}.
\endeq
\item[]\textbf{Normalization:}  As $z\to\infty$, the matrix 
${\bf N}_{m,n}(z;x)$ satisfies the condition
\begin{equation}
\label{N-large-z}
{\bf N}_{m,n}(z;x) = (\mathbb{I}+\mathcal{O}(z^{-1}))z^{n\sigma_3}
\end{equation}
with the limit being uniform with respect to direction.
\end{itemize}
\label{rhp:N}
\end{rhp}
It is immediate that  
\eq
\label{Psi-ito-N}
\Psi_n^{(m-n+1)}(0;x) = [{\bf N}_{m,n}(0;x)]_{11}
\endeq
and 
\eq
\label{H-ito-N}
\mathcal{H}_n^{(m-n+1)}(x) = -2\pi i\lim_{z\to\infty}z[{\bf N}_{m,n}(z;x)z^{-n\sigma_3}-\mathbb{I}]_{12}.
\endeq
In the next two lemmas we show how to extract $w_{m,n}^{(I)}$ and 
$w_{m,n}^{(II)}$ directly from the solution of the Riemann-Hilbert problem.  

\begin{lemma}
\label{wI-ito-N-lemma}
Write the expansion of ${\bf N}_{m,n}(z;x)$ about $z=0$ as 
\eq
\label{N-expansion}
{\bf N}_{m,n}(z;x) = {\bf N}_0(x) + {\bf N}_1(x)z + \mathcal{O}(z^2),
\endeq
where ${\bf N}_0(x)$ and ${\bf N}_1(x)$ are independent of $z$.  Then
\eq
\label{wI-ito-N}
\frac{1}{n^{1/2}}w_{m,n}^{(I)}(m^{1/2}x) = \left([{\bf N}_0(x)]_{11}[{\bf N}_0(x)]_{22}+[{\bf N}_0(x)]_{12}[{\bf N}_0(x)]_{21}-1\right)\frac{[{\bf N}_1(x)]_{11}}{[{\bf N}_0(x)]_{11}} - 2[{\bf N}_0(x)]_{12}[{\bf N}_1(x)]_{21}.
\endeq
\end{lemma}
\begin{proof}
From \eqref{Psi-ito-N}, we have 
\eq
\label{Psi-ito-N0}
\Psi_n^{(m-n+1)}(0;x)=[{\bf N}_0(x)]_{11}.
\endeq  
Thus, from the last expression in \eqref{wI-ito-log-Psi} we merely need to 
express $\frac{\partial}{\partial x}\Psi_n^{(m-n+1)}(0;x) = \left[\frac{\partial}{\partial x}{\bf N}_{m,n}(0;x)\right]_{11}$ in terms of (undifferentiated) 
entries of ${\bf N}_{m,n}$.  Define 
\eq
\label{Ntilde-def}
\widetilde{\bf N}_{m,n}(z;x):={\bf N}_{m,n}(z;x)e^{-n\theta(z;x,r)\sigma_3/2}.
\endeq
This function is analytic in $\mathbb{C}\backslash\{0\cup C\}$ with a jump 
discontinuity on $C$ that is independent of $x$ (and $z$).  This means that 
$\frac{\partial}{\partial x}\widetilde{\bf N}_{m,n}(z;x)$ has the same 
properties with the same jump on $C$.  It follows that 
\eq
\label{W-ito-Ntilde}
{\bf W}_{m,n}(z;x):=\left(\frac{\partial}{\partial x}\widetilde{\bf N}_{m,n}(z;x)\right)\widetilde{\bf N}_{m,n}(z;x)^{-1}
\endeq
is analytic in $\mathbb{C}\backslash 0$.  Inserting \eqref{Ntilde-def} into 
\eqref{W-ito-Ntilde} gives
\eq
\label{W-ito-N}
{\bf W}_{m,n}(z;x) = \left(\frac{\partial}{\partial x}{\bf N}_{m,n}(z;x)\right){\bf N}_{m,n}(z;x)^{-1} + \frac{nr^{1/2}}{z}{\bf N}_{m,n}(z;x)\sigma_3{\bf N}_{m,n}(z;x)^{-1}.
\endeq
This shows that ${\bf W}_{m,n}(z;x)$ has a simple pole at $z=0$ and, in 
particular, that $z{\bf W}_{m,n}(z;x)$ is entire in $z$.  Inserting 
the large-$z$ expansion \eqref{N-large-z} into \eqref{W-ito-Ntilde} (using 
\eqref{Ntilde-def}) shows that ${\bf W}_{m,n}(z;x)=\mathcal{O}(z^{-1})$ as 
$z\to\infty$.  This demonstrates that $z{\bf W}_{m,n}(z;x)$ is bounded as 
$z\to\infty$.  Therefore Liouville's theorem tells us that 
$z{\bf W}_{m,n}(z;x)$ is a constant matrix (i.e. independent of $z$ with 
parametric dependence on $x$).  This constant can be determined by 
considering \eqref{W-ito-N} and noting that the first summand on the 
right-hand side is bounded as $z\to 0$.  Thus
\eq
\label{W-simplified}
{\bf W}_{m,n}(z;x) = \frac{nr^{1/2}}{z}{\bf N}_{m,n}(0;x)\sigma_3{\bf N}_{m,n}(0;x)^{-1}.
\endeq
Combining \eqref{W-ito-N} and \eqref{W-simplified} gives
\eq
\frac{\partial}{\partial x}{\bf N}_{m,n}(z;x) = \frac{nr^{1/2}}{z}\left({\bf N}_{m,n}(0;x)\sigma_3{\bf N}_{m,n}(0;x)^{-1}{\bf N}_{m,n}(z;x) - {\bf N}_{m,n}(z;x)\sigma_3\right).
\endeq
Evaluating both sides at $z=0$ (using the expansion \eqref{N-expansion} on 
the right-hand side) yields
\eq
\frac{\partial}{\partial x}{\bf N}_{m,n}(0;x) = nr^{1/2}\left({\bf N}_0(x)\sigma_3{\bf N}_0(x)^{-1}{\bf N}_1(x) - {\bf N}_1(x)\sigma_3\right).
\endeq
Therefore
\eq
\label{Psi-deriv-ito-N}
\begin{split}
\frac{\partial}{\partial x}\Psi_n^{(m-n+1)}(0;x) & = \left[\frac{\partial}{\partial x}{\bf N}_{m,n}(0;x)\right]_{11} \\ 
& = nr^{1/2}\big[ \left( [{\bf N}_0(x)]_{11}[{\bf N}_0(x)]_{22} + [{\bf N}_0(x)]_{12}[{\bf N}_0(x)]_{21} - 1\right)[{\bf N}_1(x)]_{11} \\ 
 & \hspace{.55in} - 2[{\bf N}_0(x)]_{11}[{\bf N}_0(x)]_{12}[{\bf N}_1(x)]_{21} \big].
\end{split}
\endeq
Combining \eqref{wI-ito-log-Psi}, \eqref{Psi-ito-N0}, and 
\eqref{Psi-deriv-ito-N} finishes the proof.
\end{proof}

\begin{lemma}
\label{wII-ito-N-lemma}
Write the expansion of ${\bf N}_{m,n}(z;x)$ as $z\to\infty$ as 
\eq
\label{N-inf-expansion}
{\bf N}_{m,n}(z;x) = \left(\mathbb{I}+\frac{{\bf N}_{-1}(x)}{z}+\mathcal{O}\left(\frac{1}{z^2}\right)\right)z^{n\sigma_3}
\endeq
and recall the expansion \eqref{N-expansion} about $z=0$.  Then 
\eq
\frac{1}{n^{1/2}}w_{m+1,n}^{(II)}((m+1)^{1/2}x) = \left(\frac{1}{n^{1/2}}w_{m,n}^{(I)}(m^{1/2}x) + \frac{2[{\bf N}_0]_{11}[{\bf N}_0]_{12}}{[{\bf N}_{-1}]_{12}}\right)\left(1+\mathcal{O}\left(\frac{1}{m}\right)\right).
\endeq
Here $w_{m,n}^{(I)}$ can be expressed in terms of ${\bf N}_{m,n}$ via 
Lemma \ref{wI-ito-N-lemma}.  
\end{lemma}
\begin{proof}
Starting from \eqref{wII-ito-log-Psi}, we shift $m\to m+1$ and use Lemma 
\ref{wI-ito-N-lemma} to discover 
\eq
\label{wII-ito-wI-and-H}
\begin{split}
\frac{1}{n^{1/2}}w_{m+1,n}^{(II)}((m+1)^{1/2}x) & = \left(\frac{1}{n^{1/2}}\frac{\partial }{\partial x}\log\left(\Psi_n^{(m-n+1)}(0;x)\right) - \frac{1}{n^{1/2}}\frac{\partial }{\partial x}\log\left(\mathcal{H}_n^{(m-n+1)}(x)\right)\right)\frac{1}{(m+1)^{1/2}} \\
  & = \left(\frac{1}{n^{1/2}}w_{m,n}^{(I)}(m^{1/2}x) - \frac{1}{n\cdot r^{1/2}}\frac{\partial}{\partial x}\log\left(\mathcal{H}_n^{(m-n+1)}(x)\right)\right)\frac{m^{1/2}}{(m+1)^{1/2}} \\
  & = \left(\frac{1}{n^{1/2}}w_{m,n}^{(I)}(m^{1/2}x) - \frac{1}{n\cdot r^{1/2}}\frac{\frac{\partial}{\partial x}\mathcal{H}_n^{(m-n+1)}(x)}{\mathcal{H}_n^{(m-n+1)}}\right)\left(1+\mathcal{O}\left(\frac{1}{m}\right)\right).
\end{split}
\endeq
From \eqref{H-ito-N} and \eqref{N-inf-expansion} we have 
\eq
\label{H-ito-Nm1}
\mathcal{H}_{n}^{(m-n+1)}(x) = -2\pi i[{\bf N}_{-1}(x)]_{12}.
\endeq
We now express $\frac{\partial}{\partial x}\mathcal{H}_{n}^{(m-n+1)}(x) = -2\pi i\frac{\partial}{\partial x}[{\bf N}_{-1}(x)]_{12}$ in terms of 
undifferentiated entries of ${\bf N}_{m,n}$.  Insert the large-$z$ expansion 
\eqref{N-inf-expansion} into the expression \eqref{W-ito-N} for 
${\bf W}_{m,n}$:
\eq
{\bf W}_{m,n}(z;x) = \frac{1}{z}\left(\frac{\partial}{\partial x}{\bf N}_{-1}(x) + n\cdot r^{1/2}\sigma_3\right) + \mathcal{O}\left(\frac{1}{z^2}\right).
\endeq
Recalling from the proof of Lemma \eqref{wI-ito-N-lemma} that 
$z{\bf W}_{m,n}$ is a constant matrix, the $\mathcal{O}(z^{-2})$ terms must 
be identically zero.  Combining this expression with 
\eqref{W-simplified} gives
\eq
\frac{\partial}{\partial x}{\bf N}_{-1}(x) + n\cdot r^{1/2}\sigma_3 = nr^{1/2}{\bf N}_0(x)\sigma_3{\bf N}_{0}(x)^{-1}.
\endeq
Taking the (12)-entry of both sides generates
\eq
\label{Nm1x-ito-N0}
\frac{\partial}{\partial x}[{\bf N}_{-1}(x)]_{12} = -2nr^{1/2}[{\bf N}_0(x)]_{11}[{\bf N}_0(x)]_{12}.
\endeq
Using \eqref{H-ito-Nm1} and \eqref{Nm1x-ito-N0} in \eqref{wII-ito-wI-and-H} 
completes the proof of the lemma.
\end{proof}

\section{Determination of the boundary curve}
\label{sec-boundary}

We begin the Riemann-Hilbert analysis by finding the $g$-function and 
related phase function $\varphi$.  This will be sufficient to specify the 
boundary of the elliptic region, which will be used in \S\ref{sec-rhp-analysis} 
to compute the asymptotics of the rational Painlev\'e-IV functions in the 
genus-zero region.

\subsection{Construction of the $g$-function}
Suppose two complex numbers $a=a(x,r)$ and $b=b(x,r)$ are given, along with an 
oriented contour $\Sigma=\Sigma(x,r)$ from $a$ to $b$ (specifying these 
quantities is part of the process of defining the $g$-function).  The 
genus-zero $g$-function is determined via the following Riemann-Hilbert problem.
\begin{rhp}[The $g$-function] Fix $x\in\mathbb{C}$ and $r\in[1,\infty)$ and 
find $g(z)=g(z;x,r)$ such that 
\begin{itemize}
\item[]\textbf{Analyticity:}  $e^{g(z;x,r)}$ is analytic for $z\in\mathbb{C}$ 
except on $\Sigma$, where it attains H\"older-continuous boundary values at all 
interior points.  The function $g(z;x,r)$ also has a logarithmic branch cut 
that will play no role since $g$ only appears exponentiated.  
\item[]\textbf{Jump condition:}  
\eq
\label{g-jump}
g_+(z)+g_-(z) = \theta(z)+\ell, \quad z\in\Sigma
\endeq 
for some constant $\ell=\ell(x;r)$.
\item[]\textbf{Normalization:}  
\eq
\label{g-normalization}
g(z)=\log z + \mathcal{O}\left(\frac{1}{z}\right), \quad z\to\infty.
\endeq
\end{itemize}
\end{rhp}
There are some values of $x$ for which it is not possible to pick a single 
connected contour 
$\Sigma$ such that this Riemann-Hilbert problem is solvable.  When it is 
possible, then the resulting outer model problem (see Riemann-Hilbert Problem 
\ref{rhp:R} below) has jumps on a single band and the associated Riemann 
surface is genus zero.  As a result, we dub the region where the 
Riemann-Hilbert problem where $g$ is solvable the \emph{genus-zero region} 
(see Definition \ref{genus-zero-def}).  We then show that the Painlev\'e-IV 
functions are (asymptotically) free of zeros and poles in this region.  

Given $g(z)$ and $\ell$, we could define a function $\varphi$ by 
\eq
\label{phi-def}
\varphi(z;x,r)=\theta(z;x,r) - 2g(z;x,r) + \ell.
\endeq 
In actuality, we will work in the opposite order, first determining 
$\varphi'(z)$, integrating to find $\varphi(z)$, and then 
using \eqref{phi-def} to find the explicit formula for $g(z)$.  
Note $\varphi'(z)$ is specified by the following Riemann-Hilbert problem.
\begin{rhp}[The phase function $\varphi$] Fix $x\in\mathbb{C}$ and 
$r\in[1,\infty)$ and find $\varphi'(z)\equiv\varphi'(z;x,r)$ such that 
\begin{itemize}
\item[]\textbf{Analyticity:}  $\varphi'(z;x,r)$
is analytic for $z\in\mathbb{C}$ except at $z=0$ and on $\Sigma$, where it 
attains H\"older-continuous boundary values at all interior points.
\item[]\textbf{Jump condition:}  
\eq
\varphi'_+(z)+\varphi'_-(z) = 0, \quad z\in\Sigma.
\endeq 
\item[]\textbf{Pole at $z=0$:}  
\eq
\varphi'(z) = \theta'(z) + \mathcal{O}(1) = -\frac{2}{z^3}+\frac{2r^{1/2}x}{z^2}+\frac{1-r}{z} + \mathcal{O}(1), \quad z\to 0.
\label{gen0-hp-pole}
\endeq
\item[]\textbf{Normalization:}  
\eq
\varphi'(z)=-\frac{1+r}{z} + \mathcal{O}\left(\frac{1}{z^2}\right), \quad z\to\infty.
\label{gen0-hp-normalization}
\endeq
\end{itemize}
\label{rhp-hprime}
\end{rhp}
We now see how the defining relations \eqref{gen0-Q-quartic} and 
\eqref{gen0-S-ito-Q} for $Q$ and $S$ arise.  If we momentarily assume $a(x;r)$, 
$b(x;r)$, and $\Sigma$ are known, then we can define $R(z;x,r)$ by 
\eqref{R-def}.  Furthermore, writing $a+b$ as $S$ and $R(0)$ as $Q$, then we 
can see that in order to satisfy the analyticity, jump, and normalization 
conditions in Rieman-Hilbert Problem \ref{rhp-hprime}, we can choose 
$\varphi'(z)$ to have the form 
\eq
\label{phiprime-def}
\varphi'(z) = -\left((1+r)z+\frac{2}{Q}\right)\frac{R(z)}{z^3}.
\endeq
Now for $\varphi'(z)$ to satisfy the pole condition \eqref{gen0-hp-pole} at 
$z=0$, $S$ and $Q$ must satisfy the moment conditions 
\eq
\frac{(1+r)Q^3-S}{2Q^2} = -r^{1/2}x, \quad \frac{4Q^2-2(1+r)SQ^3-S^2}{8Q^4} = \frac{r-1}{2}.
\endeq
Solving the first equation for $S$ yields the relation \eqref{gen0-S-ito-Q}.  
Plugging that into the second yields the quartic equation 
\eqref{gen0-Q-quartic} for $Q$.  The specific sheet so that 
$Q(x;r)=-x+\mathcal{O}(x^{-2})$ as $x\to\infty$ is chosen so the signature 
charts in Figures \ref{fig-phase-r1} and \ref{fig-phase-r10} hold.  
Furthermore, we have assumed 
$R^2 = z^2-(a+b)z+ab = z^2-Sz+Q^2$, so we must therefore specify $a$ and $b$ 
by \eqref{ab-def}.  

We pause to indicate how the branch points of $Q(x)$ can be identified.  For 
any branch point $x_b$, the pair $\{x_b,Q(x_b)\}$ must satisify 
\eqref{gen0-Q-quartic} as well as its derivative with respect to $Q$,
\eq
\label{Q-eqn-prime}
12(1+r)^2Q^3 + 24(1+r)r^{1/2}xQ^2 + 8(r-1+rx^2)Q = 0,
\endeq
since the implicit function theorem must fail at a branch point.  Multiplying 
\eqref{Q-eqn-prime} by $Q$ gives an equation with a term proportional to 
$Q^4$.  This can be used to remove the term proportional to $Q^4$ in 
\eqref{gen0-Q-quartic}, yielding 
\eq
8(1+r)r^{1/2}xQ^3 + 8(r-1+rx^2)Q^2 - 16 = 0.
\endeq
Now \eqref{Q-eqn-prime} can be used again to remove the term proportional to 
$Q^3$, giving 
\eq
\label{Q-quadratic}
3(1+r)(rx^2+1-r)Q^2 + 2r^{1/2}x(rx^2+r-1)Q + 6(1+r) = 0.
\endeq
Now dividing \eqref{Q-eqn-prime} gives an equation that can be used to 
eliminate the term proportional to $Q^2$, yielding a linear equation for $Q$
that gives
\eq
Q = \frac{-r^2x^4+4r^2+4r+4}{2r^{1/2}(1+r)(rx^3+2(1-r)x)}.
\endeq
Plugging this into \eqref{Q-quadratic} yields the octic equation 
\eqref{xc-polynomial} for $x$ that the branch points must satisfy.  This 
equation is actually quartic in $x^2$, and so the roots can be determined 
exactly.  For $r\in[1,\infty)$, two of the roots are on the real axis, two are 
on the imaginary axis, and one is in each open quadrant.  A series expansion of 
$Q$ about the points on the axes shows that $Q$ is actually analytic there, and 
the four branch points are the ones off the axes (recall that $Q$ is also the 
solution of a quartic \eqref{gen0-Q-quartic}, and so can be written down 
explicitly to perform the series expansions).  

We return to the process of determining $\varphi$.  Now $Q$, $S$, $a$, and $b$ 
are well defined by \eqref{gen0-Q-quartic}, \eqref{gen0-S-ito-Q}, and 
\eqref{ab-def}.  So far we have seen that, for any choice of $\Sigma$, if we 
define $R$ by \eqref{R-def} then $\varphi'(z)$ must be given by 
\eqref{phiprime-def}.  The time has come to specify $\Sigma$.  
Recall the definition of $\widetilde{R}$ in \eqref{Rtilde-def}.  
Then the function $\widetilde{\varphi}(z;x,r)\equiv \widetilde{\varphi}(z)$ as 
defined in \eqref{phitilde-def} is an antiderivative of \eqref{phiprime-def} 
(with $R$ replaced with $\widetilde{R}$).  
The integration 
constant is chosen so $\widetilde{\varphi}(a)=0$.  Now 
$\Re(\widetilde{\varphi}(a))=\Re(\widetilde{\varphi}(b))$, and for $|x|$ 
sufficiently large there are two 
contours connecting $a$ and $b$ that do not pass through $z=0$ (in fact, the 
existence of both of these contours is equivalent to being in the genus-zero 
region -- see Lemma \ref{VS-is-small-lemma}).  We choose $\Sigma$ to be the 
contour connecting $a$ to $b$ when traveling clockwise around the origin.  Now 
that $\Sigma$ is 
defined, we can define $R(z)$ by \eqref{R-def} (which amounts to a deformation 
of the branch cut for $\widetilde{R}(z)$), and define 
$\varphi(z;x,r)\equiv\varphi(z)$ via 
\eq
\begin{split}
\varphi(z) := & \frac{R(z)}{Qz^2} + \left(1+r-\frac{S}{2Q^3}\right)\frac{R(z)}{z} - (1+r)\log(2z+2R(z)-S) \\ 
  & + (r-1)\log\left(\frac{2QR(z)-Sz+2Q^2}{z}\right) + \log(S^2-4Q^2) - (1+r)i\pi.
\end{split}
\endeq
Here the branches of the logarithms are chosen so $\varphi_+(z)+\varphi_-(z)=0$ 
for $z\in\Sigma$, a choice that depends on both $x$ and $r$.  
The behavior of the Riemann-Hilbert problem is controlled by $\Re(\varphi(z))$
(see Figures \ref{fig-phase-r1} and \ref{fig-phase-r10}).  
\begin{figure}[h]
\includegraphics[width=2.1in]{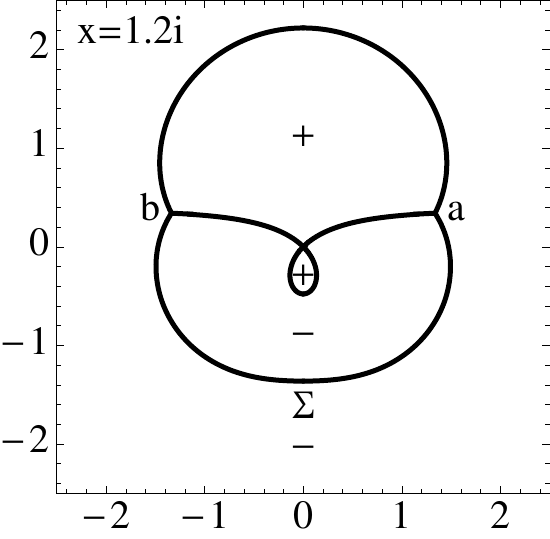} 
\includegraphics[width=2.1in]{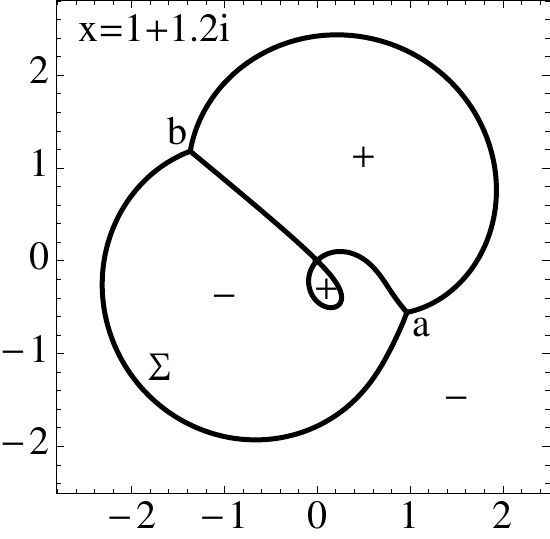} 
\includegraphics[width=2.1in]{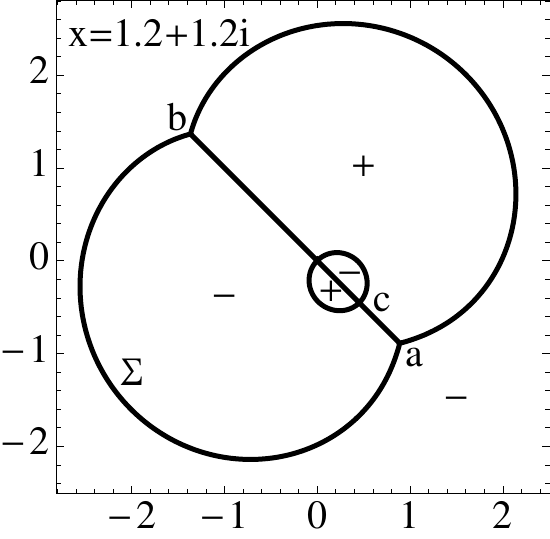} \\
\includegraphics[width=2.1in]{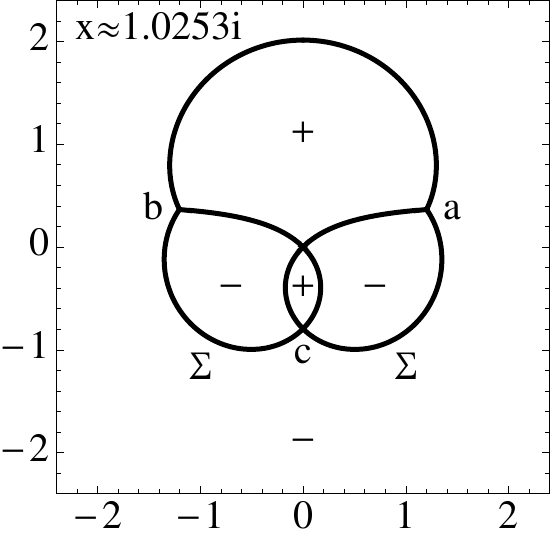}
\includegraphics[width=2.1in]{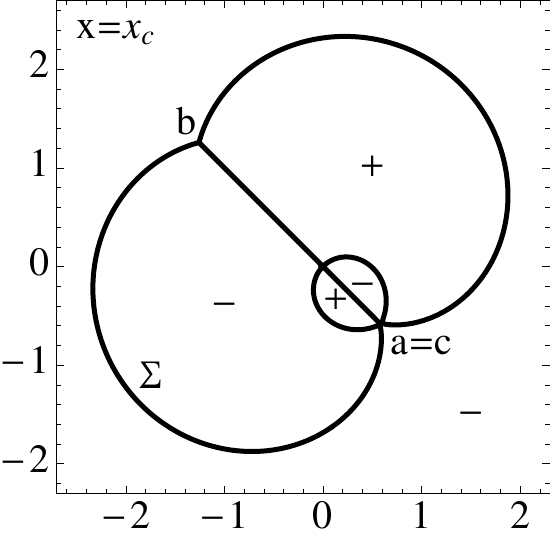}
\includegraphics[width=2.1in]{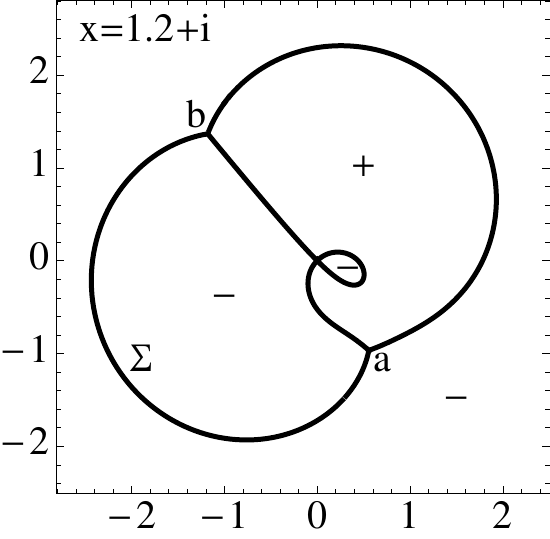} \\
\includegraphics[width=2.1in]{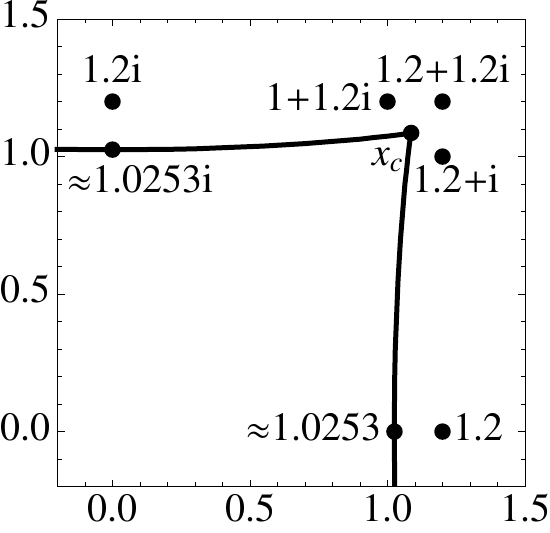}
\includegraphics[width=2.1in]{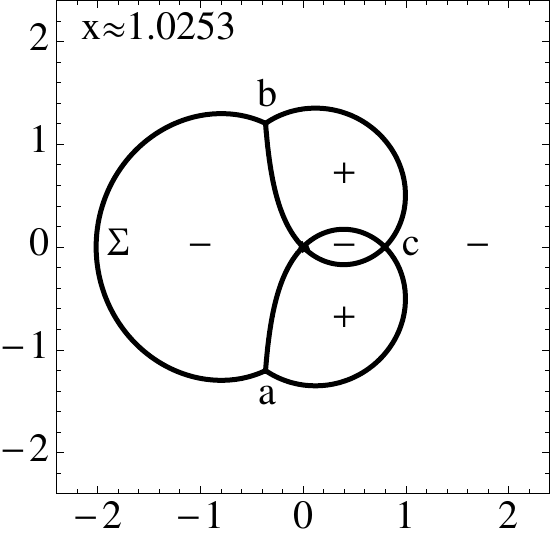}
\includegraphics[width=2.1in]{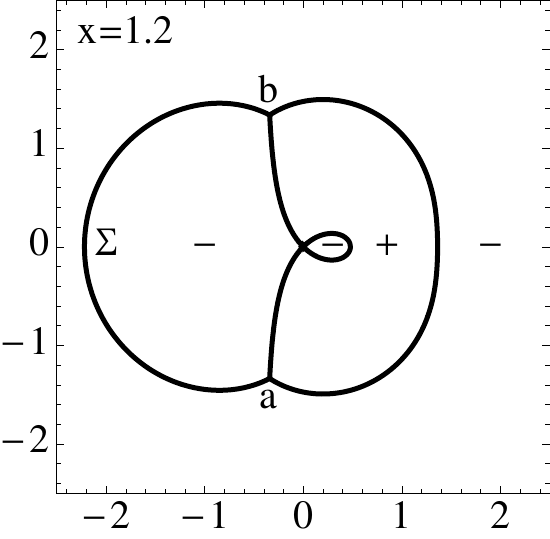}
\caption{Signature charts of $\Re(\varphi(z))$ in the complex $z$-plane with 
$r=1$ for different values of $x$ in the genus-zero region and on the boundary 
of the elliptic region.  
The band $\Sigma$ and the band endpoints $a$ and $b$ 
are indicated, as is $c$ when it lies on a zero-level line of 
$\Re(\varphi(z))$.  The topology of the zero-level lines is similar for other 
values of $r$ (see Figure \ref{fig-phase-r10} for $r=10$).  
\emph{Bottom left:}  The boundary of the elliptic region in the complex 
$z$-plane, along with the values of $x$ corresponding to the signature charts.  
}
\label{fig-phase-r1}
\end{figure}
\begin{figure}[h]
\includegraphics[width=2.1in]{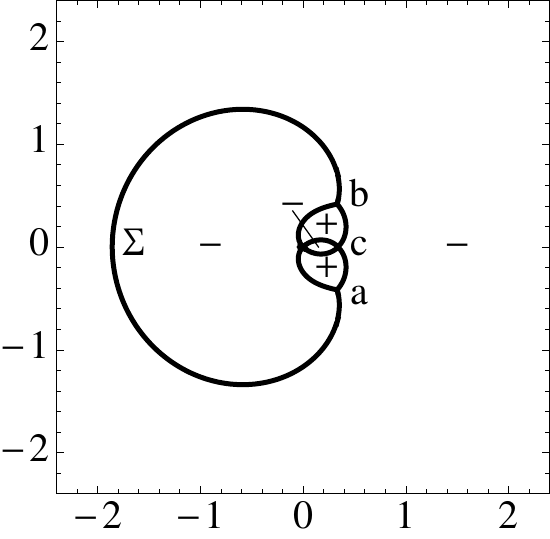}
\includegraphics[width=2.1in]{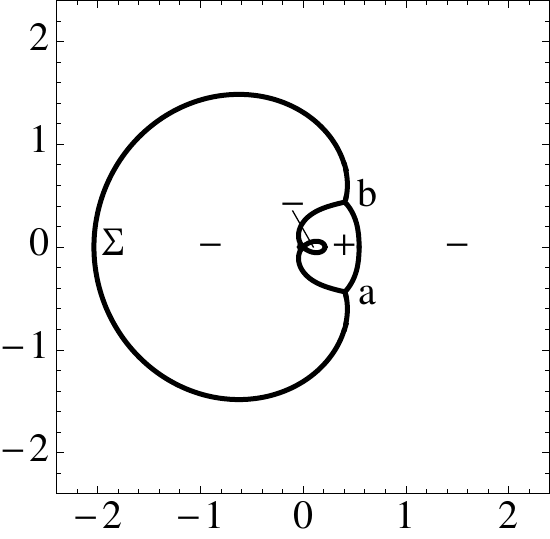}
\caption{Signature charts of $\Re(\varphi(z;x,r=10))$ in the complex $z$-plane. 
\emph{Left:}  $x\approx 1.2953$ (on the boundary of the elliptic region).  
\emph{Right:}  $x=1.4$ (in the genus-zero region).  
In both plots the band $\Sigma$ and the band endpoints $a$ and $b$ are 
indicated, as is the critical point $c$ when it lies on the zero-level line 
of $\Re(\varphi(z))$.}
\label{fig-phase-r10}
\end{figure}

Now we can set 
\eq
g(z;x,r) := \frac{1}{2}\theta(z;x,r) - \frac{1}{2}\varphi(z;x,r) + \frac{\ell(x;r)}{2},
\endeq
where only $\ell$ remains unspecified.  The role of $\ell$ is to ensure the 
normalization \eqref{g-normalization} for $g(z)$, so we choose
\eq
\ell(x;r) := 2\lim_{z\to\infty}\left(\log z - \frac{1}{2}\theta(z;x,r) + \frac{1}{2}\varphi(z;x,r)\right).
\endeq
While $\ell(x;r)$ can be computed in terms of elementary functions, we will 
not need its explicit form.

\subsection{The boundary and corners of the elliptic region}
\label{subsec-boundary}
For generic values of $x$ and $r$ the function $\varphi'(z)$ (recall 
\eqref{phiprime-def}) has three distinct zeros at 
\eq
a(x;r), \quad b(x;r), \quad \text{and} \quad c(x;r):=-\frac{2}{(1+r)Q(x;r)}.  
\label{gen0-abc}
\endeq
The transition from the genus-zero region to the elliptic region occurs when 
one of the zero-level lines of $\Re(\varphi)$ crosses $c$, i.e. 
$\Re(\varphi(c))=0$.  See the plots with $x\approx 1.0253$ and 
$x\approx 1.0253i$ in Figure \ref{fig-phase-r1} and the plot with 
$x\approx 1.2953$ in Figure \ref{fig-phase-r10}.  This condition can be written 
in the more explicit form \eqref{boundary-curve}, where $R_c=R(c)$.  It is 
important to note that the boundary of the curvilinear rectangles illustrated 
in Figures \ref{Hmn-zeros1}--\ref{Hmn-zeros3} are not the only curves along 
which $\Re(\varphi(c))=0$.  There are four additional curves that start at the 
four corners and tend to infinity (see Figure \ref{rephiatciszero}).  The 
signature chart of $\Re(\varphi(z))$ along one of these lines is illustrated 
in the plot with $x=1.2+1.2i$ in Figure \ref{fig-phase-r1}.  Nevertheless, 
the genus-zero Riemann-Hilbert analysis in \S\ref{sec-rhp-analysis} will go 
through without change along these curves, so they are part of the genus-zero 
region.  
\begin{figure}[h]
\includegraphics[width=2.1in]{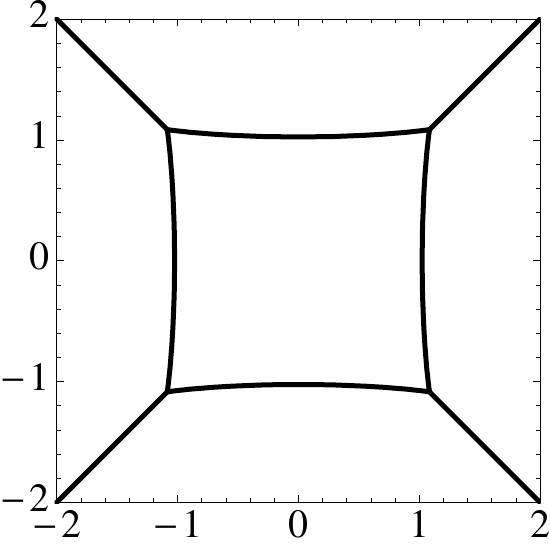}
\caption{Contours on which $\Re(\varphi(c))=0$ for $r=1$ in the complex 
$x$-plane.}
\label{rephiatciszero}
\end{figure}

\begin{remark}
As illustrated in Figure \ref{fig-phase-r1}, the breaking mechanism at the 
boundary of the elliptic region depends on whether 
$\arg(\overline{x_c})<\arg(x)<\arg(x_c)$ or 
$\arg(x_c)<\arg(x)<\arg(-\overline{x_c})$.  In the first case, a region 
in which $\Re(\varphi(z))>0$ is pinched off, as in the plot with 
$x\approx 1.0253$ in Figure \ref{fig-phase-r1}.  Looking ahead to Figure 
\ref{fig-contour-deformation}, this means it is no longer possible to pass 
the gap contour $\Gamma$ through this region in which its jump is exponentially 
close to the identity, and it is necessary to open a second band to control the 
Riemann-Hilbert problem once $x$ has moved into the elliptic region.  On the 
other hand, for $\arg(x_c)<\arg(x)<\arg(-\overline{x_c})$ (see the plot with 
$x\approx 1.0253i$ in Figure \ref{fig-phase-r1}), it is a region in which 
$\Re(\varphi(z))<0$ that is pinched off.  In this case the gap $\Gamma$ 
remains controlled, and the necessary modification occurs on the band 
$\Sigma$.  We conjecture that, as $x$ enters the elliptic region from the top 
boundary, a second band opens up directly on $\Sigma$ and then moves closer to 
the origin as $\Im(x)$ decreases.  This gives a consistent picture in which, 
just inside the boundary, there is one small and one large band.  As $x$ 
moves clockwise, the larger band rotates clockwise in the $z$-plane while the 
small band rotates counterclockwise.  The small band is near an endpoint 
of the large band exactly when $x$ is near a corner of the boundary region.  
We emphasize the Riemann-Hilbert analysis in \S\ref{sec-rhp-analysis} goes 
through uniformly for all $x$ in the genus-zero region as long as $x$ stays 
bounded away from the boundary curve.
\end{remark}

We now identify the corner points.  These are the values of $x$ for which 
$c(x)=a(x)$ or $c(x)=b(x)$ (see the plot with $x=x_c$ in Figure 
\ref{fig-phase-r1}, as well as \cite{Buckingham:2014} for a similar 
analysis for the Painlev\'e-II equation).  In either case we have 
$c^2-Sc+Q^2=0$ from 
\eqref{ab-def}.  Using \eqref{gen0-S-ito-Q} and \eqref{gen0-abc} to 
express $S$ and $c$ in terms of $Q$, $x$, and $r$ yields 
\eq
3(1+r)^2Q^4 + 4(1+r)r^{1/2}xQ^3 + 4 = 0.
\endeq
Adding this to \eqref{gen0-Q-quartic} gives
\eq
6(1+r)^2Q^4 + 12(1+r)r^{1/2}xQ^3 + 4(r-1+rx^2)Q^2 = 0,
\endeq
which is equivalent to \eqref{Q-eqn-prime}, the derivative of 
\eqref{gen0-Q-quartic} with respect to $Q$.  Once \eqref{Q-eqn-prime} holds, 
the analysis following that equation used to determine the branch points of 
$Q$ also holds, and so the corner points must satisfy \eqref{xc-polynomial}.  
While there are eight solutions to that equation, only four of them are off 
the coordinate axes, and so the geometry of the boundary shows that the 
corners are $\{\pm x_c,\pm\overline{x_c}\}$.

\section{Asymptotic expansion of the rational Painlev\'e-IV functions}
\label{sec-rhp-analysis}

We now apply the Deift-Zhou nonlinear steepest-descent method to obtain an 
approximation of ${\bf N}_{m,n}(z;x)$.  We perform a series of 
transformations
$$ {\bf N}_{m,n}(z;x) \to {\bf O}_{m,n}(z;x) \to {\bf P}_{m,n}(z;x) \to {\bf Q}_{m,n}(z;x) \approx {\bf R}_{m,n}(z;x).$$
The first transformation (to ${\bf O}_{m,n}$) deforms the jump contours away 
from the unit circle and onto $\Sigma\cup\Gamma$, where $\Gamma$ lies in a 
region where $\Re(\varphi)>0$.  The second transformation (to ${\bf P}_{m,n}$) 
introduces the $g$-function to regularize the jump matrices.  In the third 
transformation (to ${\bf Q}_{m,n}$) we open lenses, which replaces 
rapidly oscillating jump matrices with ones that are approximately constant.  
The associated Riemann-Hilbert problem is then replaced with a constant-jump 
problem that can be solved exactly for ${\bf R}_{m,n}$.  A key point is that 
the error in approximating ${\bf Q}_{m,n}$ with ${\bf R}_{m,n}$ can be 
controlled, as we will show in Lemma \ref{VS-is-small-lemma}.

\subsection{Initial deformation of the contours (${\bf N}_{m,n}\to{\bf O}_{m,n}$)}
The first step is to deform the jump contours away from the unit circle $C$.  
Define a smooth, non-self-intersecting contour $\Gamma$ starting at $b$ and 
ending at $a$ whose interior is entirely in the region in which 
$\Re(\varphi(z))>0$ (see Figure \ref{fig-contour-deformation}).  The existence 
of $\Gamma$ in the genus-zero region is shown below in Lemma 
\ref{VS-is-small-lemma}.  Then $\Sigma\cup\Gamma$ is a topological deformation 
of $C$, as shown in Figure \ref{fig-contour-deformation}.  Define $D_\text{in}$ 
to be the region in the interior of the unit circle but the exterior of 
$\Sigma\cup\Gamma$, and $D_\text{out}$ to be the region in the exterior of the 
unit circle but the interior of $\Sigma\cup\Gamma$ (again see Figure 
\ref{fig-contour-deformation}).  It is possible one of these regions may be 
empty.  Then define
\eq
{\bf O}_{m,n}(z;x) := \begin{cases} {\bf N}_{m,n}(z;x)\bbm 1 & \displaystyle\frac{1}{2\pi i}e^{-n\theta(z;x,r)} \vspace{.05in} \\ 0 & 1 \ebm, & z\in D_\text{in}, \\ {\bf N}_{m,n}(z;x)\bbm 1 & \displaystyle\frac{-1}{2\pi i}e^{-n\theta(z;x,r)} \\ 0 & 1 \ebm, & z\in D_\text{out}, \\ {\bf N}_{m,n}(z;x), & z\in\mathbb{C}\backslash\{D_\text{in}\cup D_\text{out}\}.\end{cases}
\endeq
Now ${\bf O}_{m,n}(z;x)$ satisfies exactly the same Riemann-Hilbert problem as 
${\bf N}_{m,n}(z;x)$ (i.e. Riemann-Hilbert Problem \ref{rhp:N}) with $C$ 
replaced by $\Sigma\cup\Gamma$.  
\begin{figure}[h]
\includegraphics[width=2.1in]{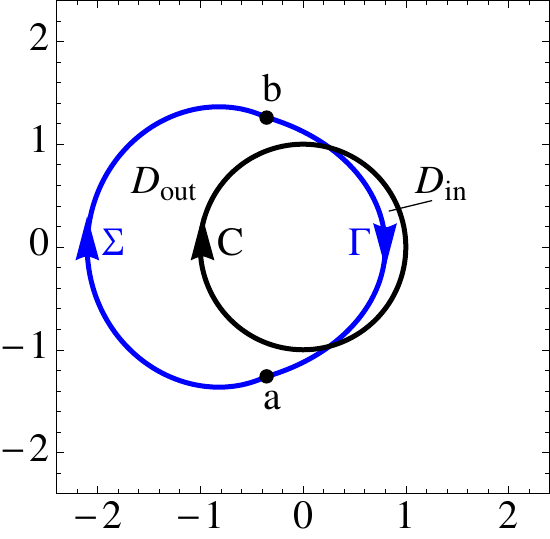}
\includegraphics[width=2.1in]{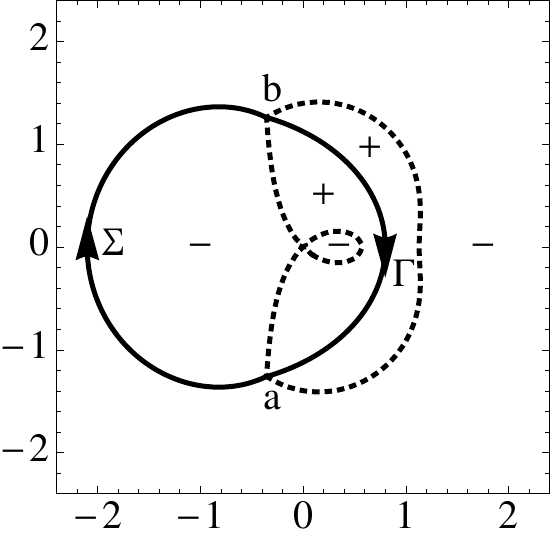}
\caption{\emph{Left:}  The contours $\Sigma$ and $\Gamma$ in relation to the 
unit circle $C$ in the complex $z$-plane for $r=1$ and $x=1.1$, along with the 
regions $D_\text{in}$ and $D_\text{out}$ used in the definition of 
${\bf O}_{m,n}(z;x)$.
\emph{Right:}  The contours $\Sigma$ and $\Gamma$ in relation to the signature 
chart of $\Re(\varphi(z))$ in the complex $z$-plane for $r=1$ and $x=1.1$.  
The contour $\Sigma$ lies 
on a zero-level line of $\Re(\varphi(z))$, while $\Gamma$ lies inside a region 
where $\Re(\varphi(z))>0$.}
\label{fig-contour-deformation}
\end{figure}

\subsection{Introduction of the $g$-function (${\bf O}_{m,n}\to{\bf P}_{m,n}$)}
Define
\eq
{\bf P}_{m,n}(z;x):=e^{-n\ell\sigma_3/2}{\bf O}_{m,n}(z;x)e^{-n(g(z;x,r)-\ell/2)\sigma_3}.
\endeq
The jump for $z\in C$ is 
\eq
{\bf V}_{m,n}^{({\bf P})} = {\bf P}_{m,n-}^{-1}{\bf P}_{m,n+} = \bbm e^{-n(g_+-g_-)} & \frac{1}{2\pi i}e^{n(g_++g_--\theta-\ell)} \\ 0 & e^{n(g_+-g_-)} \ebm.
\endeq
Recall that $\varphi(z;x,r)$ is defined in \eqref{phi-def}.  
Note from \eqref{g-jump} that $g_+(z)-g_-(z) = -\varphi_+(z) = \varphi_-(z)$ 
for $z\in\Sigma$.  Also taking into account the asymptotic behavior 
\eqref{g-normalization}, we are led to the following Riemann-Hilbert 
problem.  
\begin{rhp}[Introduction of $\varphi$]
Fix a complex number $x$ in the genus-zero region and $m,n\in\mathbb{N}$ with 
$m\geq n$ and set $r=m/n$.  
Determine the unique $2\times 2$ matrix ${\bf P}_{m,n}(z;x)$ with the 
following properties:
\begin{itemize}
\item[]\textbf{Analyticity:}  ${\bf P}_{m,n}(z;x)$
is analytic for $z\in\mathbb{C}$ except on $\Sigma\cup\Gamma$ where it 
achieves H\"older-continuous boundary values.   See Figure 
\ref{fig-contour-deformation}.  
\item[]\textbf{Jump condition:}  The boundary values taken by 
${\bf P}_{m,n}(z;x)$ are related by the jump conditions
${\bf P}_{m,n+}(z;x)={\bf P}_{m,n-}(z;x){\bf V}_{m,n}^{({\bf P})}(z;x)$, 
where
\eq
{\bf V}_{m,n+}^{({\bf P})}(z;x) = \begin{cases} \bbm e^{n\varphi_+(z;x,r)} & \displaystyle\frac{1}{2\pi i} \\ 0 & e^{n\varphi_-(z;x,r)} \ebm, & z\in\Sigma, \vspace{.05in} \\ \bbm 1 & \displaystyle\frac{1}{2\pi i}e^{-n\varphi(z;x,r)}  \\ 0 & 1 \ebm, & z\in \Gamma. \end{cases}
\endeq 
\item[]\textbf{Normalization:}  As $z\to\infty$, the matrix 
${\bf P}_{m,n}(z;x)$ satisfies the condition
\begin{equation}
{\bf P}_{m,n}(z;x) = \mathbb{I}+\mathcal{O}(z^{-1})
\end{equation}
with the limit being uniform with respect to direction.
\end{itemize}
\label{rhp:P}
\end{rhp}

\subsection{Opening of the lenses (${\bf P}_{m,n}\to{\bf Q}_{m,n}$)}
On $\Sigma$, the jump matrix ${\bf V}_{m,n}^{({\bf P})}$ has the 
factorization
\eq
\bbm e^{n\varphi_+} & \displaystyle\frac{1}{2\pi i} \\ 0 & e^{n\varphi_-} \ebm = \bbm 1 & 0 \\ 2\pi i e^{n\varphi_-} & 1 \ebm \bbm 0 & \displaystyle \frac{1}{2\pi i} \\ -2\pi i & 0 \ebm \bbm 1 & 0 \\ 2\pi i e^{n\varphi_+} & 1 \ebm.
\endeq
We introduce the lens regions $\Omega_\pm$ and the lens boundaries $L_\pm$ as 
shown in Figure \ref{fig-lenses}.
\begin{figure}[h]
\includegraphics[width=2.1in]{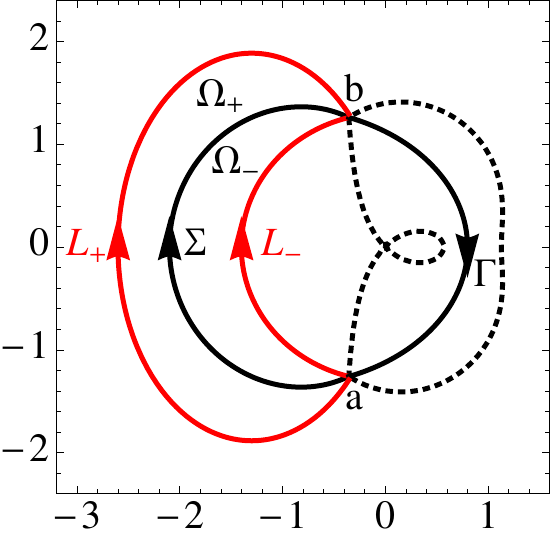}
\caption{The lens regions $\Omega_\pm$ and the lens boundaries $L_\pm$, along 
with the band $\Sigma$ and the gap $\Gamma$ in the complex $z$-plane for 
$r=1$ and $x=1.1$.  The 
zero-level lines of $\Re(\varphi)$ that are not jumps of ${\bf Q}_{m,n}$ are 
dotted.}
\label{fig-lenses}
\end{figure}
The boundaries $L_\pm$ are taken to lie inside the regions in which 
$\Re\varphi(z)<0$ and be such that $0\notin(\Omega_+\cup\Omega_-)$.  Make the 
change of variables
\eq
{\bf Q}_{m,n}(z;x):=\begin{cases} {\bf P}_{m,n}(z;x)\bbm 1 & 0 \\ -2\pi ie^{n\varphi(z;x,r)} & 1 \ebm, & z\in\Omega_+, \vspace{.05in} \\  {\bf P}_{m,n}(z;x)\bbm 1 & 0 \\ 2\pi ie^{n\varphi(z;x,r)} & 1 \ebm, & z\in\Omega_-, \\  {\bf P}_{m,n}(z;x), & \text{otherwise}.
\end{cases}
\endeq
We have the following Riemann-Hilbert problem.
\begin{rhp}[Lens-opened problem]
Fix a complex number $x$ in the genus-zero region and $m,n\in\mathbb{N}$ with 
$m\geq n$, and set $r=m/n$.  Determine the unique $2\times 2$ matrix 
${\bf Q}_{m,n}(z;x)$ with the following properties:
\begin{itemize}
\item[]\textbf{Analyticity:}  ${\bf Q}_{m,n}(z;x)$ is analytic for 
$z\in\mathbb{C}\backslash\{\Sigma\cup\Gamma\cup L_+\cup L_-\}$ with 
H\"older-continuous boundary values.  See Figure \ref{fig-lenses}.
\item[]\textbf{Jump condition:}  The boundary values taken by 
${\bf Q}_{m,n}(z;x)$ are related by the jump condition 
${\bf Q}_{m,n+}(z;x)={\bf Q}_{m,n-}(z;x){\bf V}_{m,n}^{({\bf Q})}(z;x)$, 
where
\eq
\label{Q-jumps}
{\bf V}_{m,n}^{({\bf Q})}(z;x) = \begin{cases} \bbm 0 & \displaystyle\frac{1}{2\pi i} \\ -2\pi i & 0 \ebm, & z\in\Sigma, \vspace{.05in} \\ \bbm 1 & 0 \\ 2\pi i e^{n\varphi(z;x,r)} & 1 \ebm, & z\in L_\pm, \vspace{.05in} \\ \bbm 1 & \displaystyle\frac{1}{2\pi i}e^{-n\varphi(z;x,r)}  \\ 0 & 1 \ebm, & z\in\Gamma. \end{cases}
\endeq 
\item[]\textbf{Normalization:}  As $z\to\infty$, the matrix 
${\bf Q}_{m,n}(z;x)$ satisfies the condition
\begin{equation}
{\bf Q}_{m,n}(z;x) = \mathbb{I}+\mathcal{O}(z^{-1})
\end{equation}
with the limit being uniform with respect to direction.
\end{itemize}
\label{rhp:Q}
\end{rhp}

\subsection{The model and error problems}
The jumps for ${\bf Q}_{m,n}(z)$ decay to the identity matrix except for  
$z\in\Sigma$ (although this decay is not uniform near $a$ and $b$).  We now 
define a model solution ${\bf R}_{m,n}(z)$ that is a good 
approximation for ${\bf Q}_{m,n}(z)$ (up to $\mathcal{O}(n^{-1})$) everywhere 
in the complex plane.  We begin by defining the outer model Riemann-Hilbert 
problem, which is obtained by neglecting all decaying jumps.  
\begin{rhp}[The outer model problem]
Fix a complex number $x$ in the genus-zero region and $m,n\in\mathbb{N}$ with 
$m\geq n$ and set $r=m/n$.  Determine the unique $2\times 2$ matrix 
${\bf R}_{m,n}^{\rm (out)}(z;x)$ with the following properties:
\begin{itemize}
\item[]\textbf{Analyticity:}  ${\bf R}_{m,n}^{\rm (out)}(z;x)$
is analytic in $z$ except on $\Sigma$ with H\"older-continuous boundary 
values in the interior of $\Sigma$ and at worst quarter-root singularities 
at the endpoints.  
\item[]\textbf{Jump condition:}  The boundary values taken by 
${\bf R}_{m,n}^{\rm (out)}(z;x)$ on $\Sigma$ are related by the jump condition 
\eq
{\bf R}_{m,n+}^{\rm (out)}(z;x)={\bf R}_{m,n-}^{\rm (out)}(z;x) \bbm 0 & \displaystyle \frac{1}{2\pi i} \\ -2\pi i & 0 \ebm.
\endeq 
\item[]\textbf{Normalization:}  As $z\to\infty$, the matrix 
${\bf R}_{m,n}^{\rm (out)}(z;x)$ satisfies the condition
\begin{equation}
{\bf R}_{m,n}^{\rm (out)}(z;x) = \mathbb{I}+\mathcal{O}(z^{-1})
\end{equation}
with the limit being uniform with respect to direction.
\end{itemize}
\label{rhp:R}
\end{rhp}
This constant-jump problem can be solved in a standard way by diagonalizing 
the matrix (thereby reducing the problem to two scalar problems) and then 
using the Plemelj formula.  Alternately, it is straightforward to check that 
Riemann-Hilbert Problem \ref{rhp:R} is satisfied by 
\eq
{\bf R}_{m,n}^{\rm (out)}(z;x) := \bbm \displaystyle \frac{\gamma(z;x,r)+\gamma(z;x,r)^{-1}}{2} & \displaystyle \frac{\gamma(z;x,r)-\gamma(z;x,r)^{-1}}{4\pi} \vspace{.05in}\\ \pi(\gamma(z;x,r)-\gamma(z;x,r)^{-1}) & \displaystyle \frac{\gamma(z;x,r)+\gamma(z;x,r)^{-1}}{2} \ebm,
\endeq
where
\eq
\gamma(z;x,r):=\left(\frac{z-a}{z-b}\right)^{1/4}
\endeq
is analytic for $z\notin\Sigma$ and satisfies $\lim_{z\to\infty}\gamma(z)=1$.

The outer model solution ${\bf R}_{m,n}^\text{(out)}(z)$ is a good 
approximation of ${\bf Q}_{m,n}(z)$ for all $z$ except in small 
$n$-independent neighborhoods $\mathbb{D}_a$ and $\mathbb{D}_b$ of the band 
endpoints $a$ and $b$, respectively.  Here the decay of the jumps on $L_\pm$ 
and $\Gamma$ to the identity is not uniform.  However, it is possible to 
construct functions ${\bf R}_{m,n}^{(a)}(z)$ and ${\bf R}_{m,n}^{(b)}(z)$ in 
terms of Airy functions that 
solve the Riemann-Hilbert problem exactly in their respective neighborhood and 
closely match the outer parametrix ${\bf R}_{m,n}^\text{(out)}(z)$ on the 
boundaries.  The construction of Airy parametrices is standard (see, for 
example, \cite{DeiftKMVZ:1999,Buckingham:2013}).  Here we follow 
\cite[\S4.1]{Bertola:2014}.  First, we have the local expansions 
\eq
\begin{split}
\varphi(z) & = C_a(z-a)^{3/2} + \mathcal{O}((z-a)^{5/2}), \quad z\in\mathbb{D}_a, \\
\varphi(z) & = 2\pi i + C_b(z-b)^{3/2} + \mathcal{O}((z-a)^{5/2}), \quad z\in\mathbb{D}_b,
\end{split}
\endeq
(for appropriate choices of the square roots) where $C_a$ and $C_b$ are nonzero 
and independent of $z$.  Then define two local coordinates 
\eq
s_a(z) := e^{i\pi}\left(\frac{3n}{4}\right)^{2/3}\phi(z)^{2/3} \text{ for } z\in\mathbb{D}_a; \quad \quad s_b(z) := \left(\frac{3n}{4}\right)^{2/3}(\phi(z)-2\pi i)^{2/3} \text{ for } z\in\mathbb{D}_b
\endeq
such that if $z\in\mathbb{D}_a$ then $\Gamma$ is mapped to the negative real 
axis, while  if $z\in\mathbb{D}_b$ then $\Gamma$ is mapped to the positve real 
axis.  Set ${\bf V}:=\displaystyle\frac{1}{\sqrt{2}}\bbm 1 & -i \\ -i & 1 \ebm$ 
and define the analytic prefactors
\eq
\begin{split}
{\bf B}_a(z) & := {\bf R}_{m,n}^\text{(out)}(z)(2\pi i)^{-\sigma_3/2}\bbm -i & -i \\ 1 & -1 \ebm(e^{-i\pi}s_a(z))^{\sigma_3/4}, \\ 
{\bf B}_b(z) & : = {\bf R}_{m,n}^\text{(out)}(z)(2\pi i)^{-\sigma_3/2}\bbm -i & i \\ 1 & 1 \ebm s_b(z)^{-\sigma_3/4}.
\end{split}
\endeq
Let ${\bf A}(s)$ be the function defined in 
\cite[(A.1)--(A.2)]{Bertola:2014} and built out of Airy functions with 
jumps on $\arg(s)\in\{0,\pm\frac{2\pi}{3},\pi\}$ as given in 
\cite[Figure 19]{Bertola:2014} and satisfying 
\eq
{\bf A}(s) = \frac{s^{\sigma_3/4}}{2\sqrt{\pi}}\bbm -1 & i \\ 1 & i \ebm \left(\mathbb{I}  + \frac{1}{48s^{3/2}}\bbm 1 & 6i \\ 6i & -1 \ebm + \mathcal{O}(s^{-3})\right)e^{-2s^{3/2}\sigma_3/3}, \quad s\to\infty.
\endeq
Also let $\widehat{\bf A}(s)$ be the function defined in 
\cite[(A.4)]{Bertola:2014} and built out of Airy functions with 
jumps on $\arg(s)\in\{0,\pm\frac{\pi}{3},\pi\}$ as given in 
\cite[Figure A.1]{Bertola:2014} and satisfying 
\eq
\widehat{\bf A}(s) = \frac{(e^{-i\pi}s)^{-\sigma_3/4}}{2\sqrt{\pi}}\bbm 1 & -i \\ 1 & i \ebm \left(\mathbb{I}  + \frac{i}{48s^{3/2}}\bbm -1 & 6i \\ 6i & 1 \ebm + \mathcal{O}(s^{-3})\right)e^{-2is^{3/2}\sigma_3/3}, \quad s\to\infty.
\endeq
Then the Airy parametrices are
\eq
\begin{split}
{\bf R}_{m,n}^{(a)}(z) & := i\sqrt{\pi}{\bf B}_a(z)\widehat{\bf A}(s_a(z))e^{2is_a(z)^{3/2}\sigma_3/3}(2\pi i)^{\sigma_3/2}, \quad z\in\mathbb{D}_a, \\
{\bf R}_{m,n}^{(b)}(z) & := -i\sqrt{\pi}{\bf B}_b(z){\bf A}(s_b(z))e^{2s_b(z)^{3/2}\sigma_3/3}(2\pi i)^{\sigma_3/2}, \quad z\in\mathbb{D}_b.
\end{split}
\endeq
The explicit form of the parametrix is only necessary to recover the 
$\mathcal{O}(n^{-1})$ terms in the solution of the Riemann-Hilbert problem. 
For us it suffices to know that ${\bf R}_{m,n}^{(a)}(z)$ satisfies the same 
jump conditions as ${\bf Q}_{m,n}(z)$ for $z\in\mathbb{D}_a$, 
${\bf R}_{m,n}^{(b)}(z)$ satisfies the same jump conditions as 
${\bf Q}_{m,n}(z)$ for $z\in\mathbb{D}_b$, and 
\eq
\label{Da-Db-jumps}
{\bf R}_{m,n}^{(a)}(z)={\bf R}_{m,n}^\text{(out)}(z)(\mathbb{I}+\mathcal{O}(n^{-1})) \text{ for } z\in\partial\mathbb{D}_a; \quad \quad {\bf R}_{m,n}^{(b)}(z)={\bf R}_{m,n}^\text{(out)}(z)(\mathbb{I}+\mathcal{O}(n^{-1})) \text{ for } z\in\partial\mathbb{D}_b
\endeq
uniformly for $x$ in the genus-zero region bounded away from the corners of the 
elliptic region.  At the corners one of the band endpoints collides with the 
third critical point $c$ and a different parametrix is required (see 
\cite{Buckingham:2014} for a related analysis for the rational Painlev\'e-II 
functions).  

The global model solution is now defined as 
\eq
{\bf R}_{m,n}(z;x):=\begin{cases} {\bf R}_{m,n}^\text{(out)}(z;x), & z\in\mathbb{C}\backslash\{\mathbb{D}_a\cup\mathbb{D}_b\}, \\ {\bf R}_{m,n}^{(a)}(z;x), & z\in\mathbb{D}_a, \\ {\bf R}_{m,n}^{(b)}(z;x), & z\in\mathbb{D}_b. \end{cases}
\endeq
The error or ratio function is 
\eq
{\bf S}_{m,n}(z;x):={\bf Q}_{m,n}(z;x){\bf R}_{m,n}(z;x)^{-1}.
\endeq
It satisfies the following Riemann-Hilbert problem.  Note in particular that 
${\bf S}_{m,n}(z)$ has no jump across $\Sigma$ or inside $\mathbb{D}_a$ or 
$\mathbb{D}_b$, but does have jumps across $\partial\mathbb{D}_a$ and 
$\partial\mathbb{D}_b$.  
\begin{rhp}[The error problem]
Fix a complex number $x$ in the genus-zero region and $m,n\in\mathbb{N}$ with 
$m\geq n$ and set $r=m/n$.  Determine the unique $2\times 2$ matrix 
${\bf S}_{m,n}(z;x)$ with the following properties:
\begin{itemize}
\item[]\textbf{Analyticity:}  ${\bf S}_{m,n}(z;x)$
is analytic in $z$ except on 
$J^{(S)}:=\partial\mathbb{D}_a\cup\partial\mathbb{D}_b\cup((L_+\cup L_-\cup \Gamma)\cap(\mathbb{D}_a\cup\mathbb{D}_b)^c)$ with H\"older-continuous boundary values.  We 
orient $\partial\mathbb{D}_a$ and $\partial\mathbb{D}_b$ clockwise.  See 
Figure \ref{fig-error-contours}.  
\item[]\textbf{Jump condition:}  The boundary values taken by 
${\bf S}_{m,n}(z;x)$ are related by the jump conditions 
${\bf S}_{m,n+}(z;x)={\bf S}_{m,n-}(z;x){\bf V}_{m,n}^{({\bf S})}(z;x)$, 
where
\eq
\label{S-jumps}
\begin{split}
{\bf V}_{m,n}^{({\bf S})}(z;x) & = {\bf R}_{m,n-}(z;x){\bf V}_{m,n}^{({\bf Q})}(z;x){\bf R}_{m,n+}(z;x)^{-1} \\
& = 
\begin{cases} {\bf R}_{m,n}^{\rm (out)}(z;x) \bbm 1 & 0 \\ 2\pi i e^{n\varphi(z;x,r)} & 1 \ebm {\bf R}^{\rm (out)}_{m,n}(z;x)^{-1}, & z\in L_\pm\cap(\mathbb{D}_a\cup\mathbb{D}_b)^c, \vspace{.05in} \\ {\bf R}_{m,n}^{\rm (out)}(z;x)\bbm 1 & \displaystyle\frac{1}{2\pi i}e^{-n\varphi(z;x,r)}  \\ 0 & 1 \ebm {\bf R}^{\rm (out)}_{m,n}(z;x)^{-1}, & z\in\Gamma\cap(\mathbb{D}_a\cup\mathbb{D}_b)^c, \\ {\bf R}^{\rm (a)}_{m,n}(z;x){\bf R}^{\rm (out)}_{m,n}(z;x)^{-1}, & z\in\partial\mathbb{D}_a, \\ {\bf R}^{\rm (b)}_{m,n}(z;x){\bf R}^{\rm (out)}_{m,n}(z;x)^{-1}, & z\in\partial\mathbb{D}_b. \end{cases}
\end{split}
\endeq 
\item[]\textbf{Normalization:}  As $z\to\infty$, the matrix 
${\bf S}_{m,n}(z;x)$ satisfies 
\begin{equation}
{\bf S}_{m,n}(z;x) = \mathbb{I}+\mathcal{O}(z^{-1})
\end{equation}
with the limit being uniform with respect to direction.
\end{itemize}
\label{rhp:S}
\end{rhp}
\begin{figure}[h]
\includegraphics[width=2.1in]{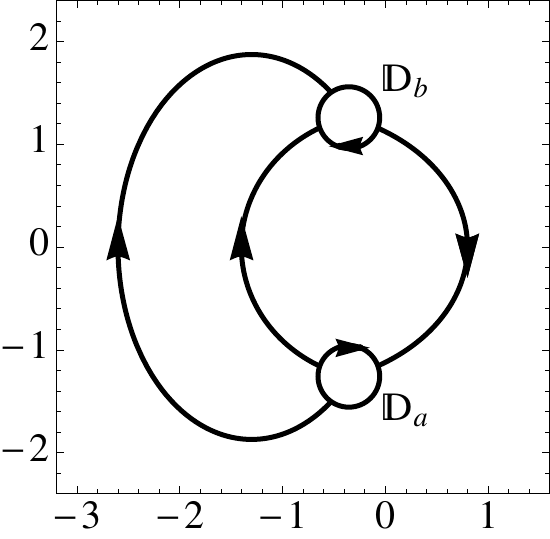}
\caption{The jump contours $J^{(S)}$ for the error problem ${\bf S}_{m,n}(z;x)$ 
in the complex $z$-plane for $r=1$ and $x=1.1$.}
\label{fig-error-contours}
\end{figure}
We now show that the jump matrices for the error solution ${\bf S}_{m,n}$ are 
small as $n\to\infty$.  
\begin{lemma}
\label{VS-is-small-lemma}
Fix $\delta>0$.  Then for $z\in J^{(S)}$
\eq
\label{VS-is-small}
{\bf V}_{m,n}^{({\bf S})}(z;x) = \mathbb{I} + \mathbb{O}\left(\frac{1}{n}\right)
\endeq
with the error term uniform in $x$ if ${\rm dist}(x,E_r)>\delta$.
\end{lemma}
\begin{proof}
For $z\in\partial\mathbb{D}_a\cup\partial\mathbb{D}_b$, the necessary 
estimate is given by \eqref{Da-Db-jumps}.  What remains is to show that, in 
the genus-zero region, the signature chart of $\Re(\varphi(z))$ has the 
topology shown in the genus-zero plots in Figures \ref{fig-phase-r1} and 
\ref{fig-phase-r10}.  More specifically, we 
need to show that (except for the endpoints $a$ and $b$), $L_\pm$ can be placed 
entirely in a region in which $\Re(\varphi(z))<0$, and $\Gamma$ entirely in a 
region in which $\Re(\varphi(z))>0$.  If so, then with these choices we find 
that ${\bf V}_{m,n}^{({\bf S})}(z;x)$ is exponentially close to the 
identity on the relevant parts of $L_\pm$ and $\Gamma$, and so 
\eqref{VS-is-small} holds.  

As a level set of a function that is harmonic except on $\Sigma$ and at $z=0$, 
$\{z:\Re(\varphi(z))=0\}$ consists of a finite number of smooth arcs.  Local 
analysis at infinity shows there are no zero-level lines of $\Re(\varphi(z))$ 
there.  The only points at which two or more zero-level lines can intersect 
are the critical points $a$, $b$, and $c$ (see \eqref{gen0-abc}) or the 
origin.  A direct calculation shows that $a$ and $b$ are distinct and nonzero.  
We also saw in \S\ref{subsec-boundary} that $c$ can coincide with $a$ or $b$, 
but only at the corners of the elliptic region, which we avoid.  Furthermore, 
$c$ cannot be zero since $Q$ has no finite singularities.  Therefore, we can 
assume all four points $a$, $b$, $c$, and 0 are distinct.  By construction, 
$\Re(\varphi(a))=\Re(\varphi(b))=0$, and local analysis shows there are three 
zero-level lines of $\Re(\varphi(z))$ emanating from both $a$ and $b$.  
Similarly, local analysis at the pole $z=0$ shows there are four zero-level 
lines of $\Re(\varphi(z))$ intersecting at the origin.  As we have seen in 
\S\ref{subsec-boundary}, $\Re(\varphi(c))$ is generically nonzero, but there 
are four semi-infinite arcs in the complex $x$-plane along which 
$\Re(\varphi(c))=0$, in which case four zero-level lines of $\Re(\varphi(z))$ 
intersect at $c$.  

First, assume $x$ is such that $\Re(\varphi(c))\neq 0$.  In this case we have 
three arcs each emerging from $a$ and $b$ and four from 0.  Therefore not all 
the arcs from $a$ and $b$ can connect to the origin, and at least one must 
join $a$ and $b$.  Closed contours that are level lines of harmonic functions 
must enclose singularities, and so there are two options:  either a second 
arc connects $a$ to $b$ and passes around the opposite side of the origin from 
the first such arc, or the other four arcs from $a$ and $b$ all connect to the 
origin.  We are in the first situation for $x$ sufficiently large and either 
$x$ purely real or purely imaginary (for illustration see the plots with 
$x=1.2$ and $x=1.2i$ in Figure \ref{fig-phase-r1}).  In this case the signature 
chart necessarily has the form show in those plots since $\Re(\varphi(z))<0$ 
for $z$ sufficiently large.  The only allowable mechanism for the contour 
topology to change as $x$ varies is for $c$ to intersect a zero-level line 
of $\Re(\varphi(z))$, which we have seen only occurs on the semi-infinite arcs. 
Therefore, off these four arcs we see that the contours $L_\pm$ and $\Gamma$ 
can be chosen appropriately in the exterior of the elliptic region.

We now consider $x$ such that $\Re(\varphi(c))=0$.  The signature chart at a 
corner point can be seen to have the form shown in the plot with $x=x_c$ in 
Figure \ref{fig-phase-r1}.  It is possible to continuously vary $x$ to the 
value that interests us keeping $c$ on a zero-level curve of $\Re(\varphi(z))$. 
Therefore the signature chart of $\Re(\varphi(z))$ must (topologically) have 
the form illustrated in the plot with $x=1.2+1.2i$ in Figure 
\ref{fig-phase-r1}, from which it is clear the contours $L_\pm$ and $\Gamma$ 
can be chosen as needed.
\end{proof}

We have finally arrived at a small-norm Riemann-Hilbert problem for 
${\bf S}_{m,n}(z;x)$, that is, one with jumps close to the identity.  The 
following analysis is standard (see, for example, \cite{DeiftZ:1993} or 
\cite[Appendix B]{Buckingham:2013}).  Recursively define the functions 
\eq
{\bf U}_0(z):=\mathbb{I}, \quad {\bf U}_k(z):=\frac{-1}{2\pi i}\int_{J_-^{(S)}}\frac{{\bf U}_{k-1}(u)({\bf V}_{m,n}^{({\bf S})}(u)-\mathbb{I})}{z-u}du,
\endeq
in which $J_-^{(S)}$ means the integration is performed along the minus-side of 
$J^{(S)}$.  Then ${\bf S}_{m,n}(z)$ is the sum of an infinite Neumann series:
\eq
{\bf S}_{m,n}(z) = \mathbb{I} - \frac{1}{2\pi i}\sum_{k=1}^\infty \int_{J^{(S)}}\frac{{\bf U}_{k-1}(u)({\bf V}_{m,n}^{({\bf S})}(u)-\mathbb{I})}{z-u}du.
\endeq
This gives us the bound
\eq
{\bf S}_{m,n}(z) = \left(\mathbb{I}+\mathcal{O}\left(\frac{1}{(|z|+1)n}\right)\right), \quad n\to\infty
\endeq
that holds uniformly for $z\in\mathbb{C}\backslash J^{(S)}$ and for $x$ a fixed 
distance away from the elliptic region.

\subsection{The asymptotic expansion}
We now prove the main theorems.
\begin{proof}[Proof of Theorems \ref{thm-wI}, \ref{thm-wII}, and \ref{thm-wIII}]
Retracing the various transformations gives
\eq
{\bf N}_{m,n}(z;x) = e^{n\ell\sigma_3/2}{\bf S}_{m,n}(z;x){\bf R}_{m,n}(z;x)e^{n(g(z;x,r)-\ell/2)\sigma_3}, \quad z\in\mathbb{C}\backslash\{\Omega_+\cup\Omega_-\cup D_\text{in} \cup D_\text{out}\}.
\endeq
We therefore have 
\eq
\label{N-unraveled}
\begin{split}
{\bf N}_{m,n}(z) = \left(\mathbb{I}+\mathcal{O}\left(\frac{1}{(|z|+1)n}\right)\right) \bbm \displaystyle \frac{\gamma(z)+\gamma(z)^{-1}}{2}e^{ng(z)} & \displaystyle \frac{\gamma(z)-\gamma(z)^{-1}}{4\pi}e^{-n(g(z)-\ell)} \vspace{.05in}\\ \pi(\gamma(z)-\gamma(z)^{-1})e^{n(g(z)-\ell)} & \displaystyle \frac{\gamma(z)+\gamma(z)^{-1}}{2}e^{-ng(z)} \ebm, \\
z\in\mathbb{C}\backslash\{\Omega_+\cup\Omega_-\cup D_\text{in} \cup D_\text{out} \cup \mathbb{D}_a \cup \mathbb{D}_b \cup J^{(S)}\}.
\end{split}
\endeq
In particular, this expression holds for $z=0$ and for $|z|$ sufficiently 
large.  We expand $g(z)$ and $\gamma(z)$ about $z=0$:
\eq
g(z;x,r) = g_0(x,r) + g_1(x,r)z + \mathcal{O}(z^2), \quad \gamma(z;x,r) = \gamma_0(x,r) + \gamma_1(x,r)z + \mathcal{O}(z^2),
\endeq
wherein
\eq
\label{gamma0-gamma1}
\gamma_0 = \left(\frac{a}{b}\right)^{1/4}, \quad \gamma_1 = \frac{a-b}{4ab}\left(\frac{a}{b}\right)^{1/4}
\endeq
(interestingly, it will turn out that we will not need the explicit form of 
$g_0$ or $g_1$).  Thus, recalling the expansion \eqref{N-expansion} for 
${\bf N}_{m,n}$, we compute
\eq
\label{N0-N1-entries}
\begin{gathered}[]
[{\bf N}_0]_{11} = \frac{\gamma_0+\gamma_0^{-1}}{2}e^{ng_0}(1+\mathcal{O}(n^{-1})), \quad [{\bf N}_0]_{12} = \frac{\gamma_0-\gamma_0^{-1}}{4\pi}e^{-n(g_0-\ell)}(1+\mathcal{O}(n^{-1})), \\
[{\bf N}_0]_{21} = \pi(\gamma_0-\gamma_0^{-1})e^{n(g_0-\ell)}(1+\mathcal{O}(n^{-1})), \quad [{\bf N}_0]_{22} = \frac{\gamma_0+\gamma_0^{-1}}{2}e^{-ng_0}(1+\mathcal{O}(n^{-1})), \\
[{\bf N}_1]_{11} = \frac{1}{2}\left[(\gamma_0+\gamma_0^{-1})ng_1 + \gamma_1 - \frac{\gamma_1}{\gamma_0^2}\right]e^{ng_0}(1+\mathcal{O}(n^{-1})), \\
[{\bf N}_1]_{21} = \pi\left[(\gamma_0-\gamma_0^{-1})ng_1 + \gamma_1 + \frac{\gamma_1}{\gamma_0^2}\right]e^{n(g_0-\ell)}(1+\mathcal{O}(n^{-1})).
\end{gathered}
\endeq
Inserting these into \eqref{wI-ito-N} gives
\eq
\frac{1}{n^{1/2}}w_{m,n}^{(I)}(m^{1/2}x) = 2\frac{\gamma_1(x,r)}{\gamma_0(x,r)}\frac{(1-\gamma_0(x,r)^2)}{(1+\gamma_0(x,r)^2)} + \mathcal{O}(n^{-1}).
\endeq
Using the expressions \eqref{gamma0-gamma1} for $\gamma_0$ and $\gamma_1$ 
produces
\eq
\frac{1}{n^{1/2}}w_{m,n}^{(I)}(m^{1/2}x) = \frac{1}{(a(x,r)b(x,r))^{1/2}} - \frac{a(x,r)+b(x,r)}{2a(x,r)b(x,r)} + \mathcal{O}(n^{-1}).
\endeq
Finally, using the identities $S=a+b$ and $Q=-(ab)^{1/2}$ gives 
\eqref{wI-ito-Q-and-S} in the genus-zero region.  This completes the proof of 
Theorem \ref{thm-wI}.  See also Figure \ref{wI-approx-plots}.

Next, we compute the asymptotic expansion of $w_{m,n}^{(II)}$, starting 
from Lemma \ref{wII-ito-N-lemma}.  From \eqref{N-unraveled} we have 
\eq
[{\bf N}_{m,n}(z)]_{12} = \frac{\gamma(z)-\gamma(z)^{-1}}{4\pi}e^{-n(g(z)-\ell)}\left(1+\mathcal{O}\left(n^{-1}\right)\right)
\endeq
for $x$ in the genus-zero region.  We expand $\gamma(z)$ at infinity as 
\eq
\gamma(z) = 1 + \frac{b-a}{4z} + \mathcal{O}\left(\frac{1}{z^2}\right).
\endeq
Using the last two equations along with $g(z)=\log(z)+\mathcal{O}(z^{-1})$ 
and the expansion \eqref{N-inf-expansion} shows 
\eq
[{\bf N}_{-1}]_{12} = \frac{(b-a)e^{n\ell}}{8\pi}(1+\mathcal{O}(n^{-1})).
\endeq
Taking this along with \eqref{N0-N1-entries} and then \eqref{gamma0-gamma1} 
shows
\eq
\frac{[{\bf N}_0]_{11}[{\bf N}_0]_{12}}{[{\bf N}_{-1}]_{12}} = \frac{\gamma_0^2-\gamma_0^{-2}}{b-a} = -\frac{1}{(ab)^{1/2}} = \frac{1}{Q}.
\endeq
We now plug this and \eqref{wI-ito-Q-and-S} into the result of Lemma 
\ref{wII-ito-N-lemma} to see 
\eq
\frac{1}{n^{1/2}}w_{m+1,n}^{(II)}((m+1)^{1/2}x) = \left( \frac{1}{Q(x,r)} - \frac{S(x,r)}{2Q(x,r)^2} + \mathcal{O}\left(n^{-1}\right) \right)\left(1+\mathcal{O}\left(m^{-1}\right)\right).
\endeq
As long as we agree $r=\frac{m}{n}$ is fixed, we can replace 
$\mathcal{O}(m^{-1})$ with $\mathcal{O}(n^{-1})$.  Therefore, we have 
\eq
\frac{1}{n^{1/2}}w_{m,n}^{(II)}(m^{1/2}x) = \frac{1}{Q(x,\frac{m-1}{n})} - \frac{S(x,\frac{m-1}{n})}{2Q(x,\frac{m-1}{n})^2} + \mathcal{O}\left(n^{-1}\right).
\endeq
From the dependence of $Q(x,r)$ and $S(x,r)$ on $r$, we can replace 
$Q(x,\frac{m-1}{n})$ and $S(x,\frac{m-1}{n})$ with $Q(x,r)$ and $S(x,r)$, 
respectively, at the price of an $\mathcal{O}(n^{-1})$ error, so we obtain 
our final result \eqref{wII-ito-Q-and-S}.  This completes the proofs of 
Theorems \ref{thm-wII} and \ref{thm-wIII} (as Theorem \ref{thm-wIII} follows 
immediately from Theorems \ref{thm-wI} and \ref{thm-wII}).  See Figure 
\ref{wII-approx-plots} for plots demonstrating the convergence for 
$w_{m,n}^{(II)}$.
\end{proof}

\end{document}